\documentclass[letter,11pt]{article}
\usepackage{makeidx}  
\usepackage{amssymb}  
\usepackage[dvips]{graphicx}
\usepackage{amsmath}
\usepackage{amsthm}
\usepackage[section]{algorithm}
\usepackage{algorithmic}
\usepackage{authblk}
\newcommand{\Order}{\mathcal{O}}

\newcommand{\z}{z}
\newcommand{\PHC}{\mathrm{PHC}}

\renewcommand{\Vec}[1]{\mbox{\boldmath $#1$}}
\newcommand{\visit}{{\rm visit}}
\newcommand{\PP}{$\mathcal{P}$}
\newcommand{\NP}{$\mathcal{NP}$}

\newtheorem{dfn}{Definition}[section]
\newtheorem{thm}[dfn]{Theorem}
\newtheorem{lem}[dfn]{Lemma}

\newtheorem{corollary}[dfn]{Corollary}
\newtheorem{prop}[dfn]{Proposition}

\newtheorem{alg}{Algorithm}
\newtheorem{problem}{Problem}

\title{The Parity Hamiltonian Cycle Problem}
\author[1]{Hiroshi Nishiyama}
\author[2]{Yusuke Kobayashi}
\author[1]{Yukiko Yamauchi}

\author[1]{\\Shuji Kijima}
\author[1]{Masafumi Yamashita}
\affil[1]{
  Graduate School of Information Science and Electrical Engineering, 

  Kyushu University, Fukuoka, 819-0395, Japan

{\ttfamily \{hiroshi.nishiyama, yamauchi, kijima, mak\}@inf.kyushu-u.ac.jp}
}
\affil[2]{
Faculty of Engineering, Information and Systems, 

University of Tsukuba, Tsukuba, Ibaraki, 305-8573, Japan

{\ttfamily  kobayashi@sk.tsukuba.ac.jp}
}

\begin{document}
\maketitle

\begin{abstract}
 Motivated by a relaxed notion of the celebrated {\em Hamiltonian cycle}, 
  this paper investigates its variant, 
   {\em parity Hamiltonian cycle} ({\em PHC}\/): 
 A PHC of a graph is a closed walk 
  which visits every vertex an odd number of times,
  where we remark that the walk may use an edge more than once. 
 First, we give a complete characterization of the graphs which have PHCs, and 
  give a linear time algorithm to find a PHC,  
   in which every edge appears at most four times, in fact. 
 In contrast, 
  we show that finding a PHC is {\NP}-hard 
  if a closed walk is allowed to use each edge at most $z$ times for each $z=1,2,3$ ($\PHC_z$ for short), 
  even when a given graph is two-edge connected.  
 We then further investigate the $\PHC_3$ problem, and
  show that the problem is in {\PP} when an input graph is four-edge connected. 
 Finally,
  we are concerned with three (or two)-edge connected graphs, and
 show that the $\PHC_3$ is in {\PP} 
  for any $C_{\geq 5}$-free or $P_6$-free graphs.
 Note that the Hamiltonian cycle problem 
  is known to be {\NP}-hard for those graph classes.

\end{abstract}

\section{Introduction}
 It is said that
  the graph theory has its origin in the seven bridges of
  K\"{o}nigsberg settled by Leonhard Euler~\cite{BLW}.
 An {\em Eulerian} cycle, named after him in modern terminology,
  is a cycle which uses every edge exactly once, and  
  now it is well-known that 
  a connected undirected graph has an Eulerian cycle
  if and only if every vertex has an even degree.
 A {\em Hamiltonian} cycle (HC), a similar but completely different
notion, is a cycle which visits every vertex exactly once.
 In contrast to the clear characterization of an Eulerian graph,
   the question if a given graph has a Hamiltonian cycle
   is a celebrated NP-complete problem due to Karp~\cite{Karp}.
 The HC problem is widely interested in computer science or mathematics,
  and has been approached with several variant or related problems.
 The traveling salesman problem (TSP) in a graph 
  is a problem to minimize the length of a walk 
  which visits every vertex at least once, instead of exactly once. 
 Another example may be a two-factor (in cubic graphs), 
  which is a collection of cycles covering every vertex exactly once, 
  meaning that the connectivity of the HC is relaxed~(cf.~\cite{Hartvigsen,HL11,BSvS,BIT13}).

 It could be a natural idea for the HC problem
  to modify the condition on the visiting number keeping the connectivity condition.
 This paper proposes 
  the {\em parity Hamiltonian cycle} ({\em PHC}\/), which is
  a variant of the Hamiltonian cycle:
  a PHC is a closed walk 
  which visits every vertex an odd number of times (see Section~\ref{sec:prelim},  for more rigorous description). 
 Note that 
  a closed walk is allowed to use an edge more than once. 
 The PHC problem is to decide if a given graph has a PHC. 
 We remark that 
  another version of the problem which is to find a closed walk visiting each vertex an even number of times is trivial: 
  find a spanning tree and trace it twice.

 \begin{table}[t]
	\begin{center}
		\caption{Time complexity of the PHC problem.}
		\label{tab:PHC-complexity}
		{\renewcommand\arraystretch{1.1}
		 \tabcolsep = 2mm
		\begin{tabular}{|c|c@{~}l|} \hline
			Each edge is used at most $\z$ times& Complexity &\\ \hline\hline
			$\z \geq 4$ & {\PP} &(Thm.~\ref{thm:PHC4_P}) \\ \hline
			$\z = 3$ & {\NP}-complete &(Thm.~\ref{thm:PHC3_NPC}) $\Rightarrow$ see Table~\ref{tab:PHC3-complexity-edge-connectivity}\\ \hline 
			$\z = 2$ & {\NP}-complete &(Thm.~\ref{thm:PHC2_NPC})\\ \hline 
			$\z = 1$ & {\NP}-complete &(Thm.~\ref{thm:PHC1_NPC})\\ \hline
		\end{tabular}
		}
	\end{center}
\end{table}
 \begin{table}[t]
	\begin{center}
		\caption{Time complexity of the $\PHC_3$ problem.}
		\label{tab:PHC3-complexity-edge-connectivity}
		{
		 \tabcolsep = 2mm
		\begin{tabular}{|l|c|c|} \hline
			Edge connectivity & Complexity  & Refinement by graph classes \\ \hline\hline
			4-edge connected & \multicolumn{1}{c}{\PP~(Thm.~\ref{thm:PHC3-4EC})} &  \\ \hline
			3-edge connected & unknown  & See Fig.~\ref{fig:hierarchy_of_classes} \\ \cline{1-2}
			2-edge connected & {\NP}-complete (Thm.~\ref{thm:PHC3_NPC})& (Section~\ref{sec:PHC3})   \\ \hline
			1-edge connected & {\NP}-complete& Reduced to 2-edge connected\\
&   (from 2-edge connected case)& (Prop.~\ref{prop:nec_PHC_bridge},~\ref{prop:suf_PHC_bridge})\\ \hline
		\end{tabular}
		}
	\end{center}
 \end{table}

 It may not be trivial if the PHC problem is in {\NP}, 
  since the length of a PHC is unbounded in the problem. 
 This paper firstly shows in Section~\ref{sec:easy4} that
  the PHC problem is in {\PP}, in fact. 
 Precisely, 
  we give a complete characterization of the graphs which have PHCs. 
 Furthermore, we show that 
  if a graph has a PHC then we can find a $\PHC_4$ in linear time, 
  where $\PHC_z$ for a positive integer~$z$ denotes a PHC which uses each edge at most $z$ times. 
 In contrast, Section~\ref{sec:hard123} shows that the ${\rm PHC}_z$ problem is {\NP}-complete 
  for each $z=1,2,3$ (see Table~\ref{tab:PHC-complexity})\footnote{ 
 Notice that 
   those hardness results are independent, 
   e.g., ``$z=3$ is hard'' does not imply ``$z=2$ is hard,'' and vice versa. 
}. 
 We then further investigate the $\PHC_3$ problem. 
 In precise, the $\PHC_3$ problem is {\NP}-complete for {\em two}-edge connected graphs, 
  while Section~\ref{sec:4-edge-connected} shows that it is solved in polynomial time for {\em four}-edge connected graphs. 
 The complexity of the $\PHC_3$ for {\em three}-edge connected graphs 
  remains unsettled (see Table~\ref{tab:PHC3-complexity-edge-connectivity}).

 As an approach to the ${\rm PHC}_3$ problem for three-edge connected graphs, 
  we in Section~\ref{sec:PHC3} utilize the celebrated {\em ear-decomposition}, 
  which is actually a well-known characterization of {\em two}-edge connected graphs. 
 Then, Section~\ref{sec:PHC3} shows that the $\PHC_3$ problem is in \PP\/ 
  for any two-edge connected $C_{\geq 5}$-free or $P_6$-free graphs 
  (see Section~\ref{sec:pre-graph} for the graph classes). 
 The classes of $C_{\geq 5}$-free or $P_6$-free contain some important graph classes 
  such as chordal, chordal bipartite and cograph. 
 We remark that it is known that the Hamiltonian cycle is {\NP}-complete 
  for $C_{\geq 4}$-free graphs, as well as $P_5$-free graphs~(cf.~\cite{BLS87}). 

 In precise, 
  we in Section~\ref{sec:PHC3} introduce a stronger notion of {\em all-roundness} 
  (and {\em bipartite all-roundness}) of a graph, 
  which is a sufficient condition that a graph has a PHC. 
 Catlin~\cite{Catlin} presented a similar notion of {\em collapsible} 
  in the context of spanning Eulerian subgraphs, and 
  the all-roundness is an extended notion of the collapsible. 
 Then, we show that 
  any two-edge connected $C_{\geq 5}$-free or $P_6$-free graphs are all-round or bipartite all-round. 
 We conjecture that 
  any two-edge connected $C_{\geq 7}$-free graphs are all-round, 
  while it seems not true for $C_{\geq 8}$-free nor $P_7$-free.

 \begin{figure}[tbp]
    \begin{center}
        \includegraphics[width=90mm]{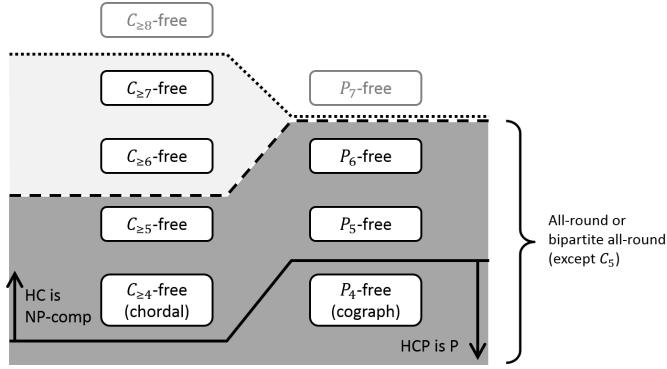}
     \end{center}
        \caption{$\PHC_3$ is in \PP\/ for $C_{\geq 5}$-free graphs and $P_6$-free graphs. 
 \label{fig:hierarchy_of_classes}}
 \end{figure}
 Section~\ref{sec:discussion} is for a miscellaneous discussion. 
 Section~\ref{sec:brdge} 
  extends the results for two-edge connected graphs in Section~\ref{sec:PHC3} to (one-edge) connected graphs. 
 In Section~\ref{sec:dense}, we remark that a dense graph is also all-round 
  using some techniques in Section~\ref{sec:PHC3}. 
 Before closing the paper, 
  Section~\ref{sec:connection} briefly discusses the connection 
  between the PHC and other problems, such as Hamiltonian cycle or Eulerian cycle, 
  regarding a generalized problem described in Section~\ref{sec:PHC3}.


 \paragraph{Related works}
 Here, we refer to the work by Brigham et al.~\cite{Brigham}, 
  which investigated a similar (or, essentially the same) problem. 
 Brigham et al.~\cite{Brigham}
  showed that any connected undirected graph has a parity Hamiltonian path or cycle 
  (see Section~\ref{sec:Brigham}). 
 Their proof was constructive, and 
  they gave a linear time algorithm based on the depth first search.
 As far as we know, 
  it is the uniquely known result on the problem\footnote{
    Very recently, we have investigated a PHC on directed graphs in~\cite{Nishiyama2} after this paper.  
}. 
 To be precise, 
  we remark that their argument does not imply that the PHC problem is in \PP: 
  for the purpose, we need an argument when a graph does not contain a parity Hamiltonian {\em cycle}. 

 Notice that
  the condition that a HC visits each vertex $1 \in \mathbb{R}$ times
  is replaced by $1 \in \mbox{GF(2)}$ times in the PHC.
 Modification of the field is found in some graph problems, 
  such as group-labeled graphs or nowhere-zero flows~\cite{JND12,Kochol02}.
 It was shown that the extension complexity of the TSP is
exponential~\cite{Yannakakis,FMP12,FMP15},
  while it is an interesting question if the PHC has an efficient
(extended) formulation over GF(2).

\section{Preliminary} \label{sec:prelim}
 This section introduces graph terminology of the paper. 
 First, we define the {\em parity Hamiltonian cycle} in Section~\ref{sec:pre-phc},  
  with an almost minimal introduction of a graph terminology. 
 Then, we in Section~\ref{sec:pre-graph} give some other terminology, 
  such as $T$-join, ear decomposition, graph classes, 
  which are commonly used in the graph theory. 
\subsection{Parity Hamiltonian cycle}\label{sec:pre-phc}
\subsubsection{Definition}
 An \emph{undirected simple graph} (simply we say ``{\em graph}'') $G = (V, E)$ is given by 
  a vertex set $V$ (or we use $V(G)$) and 
  an edge set $E \subseteq {V \choose 2}$ (or $E(G)$). 
 Let $\delta_G(v)$ denote the set of incident edges to $v$, and  
 let $d_G(v)$ denote the degree of a vertex $v$ in $G$, i.e., $d_G(v)=|\delta_G(v)|$. 
 We simply use $\delta(v)$ and $d(v)$ without a confusion.

 A {\em walk} is an alternating sequence of vertices and edges 
    $v_0 e_1 v_1 \cdots v_{\ell-1} e_{\ell} v_{\ell}$ 
   with an appropriate $\ell \in \mathbb{Z}_{\geq 0}$,
    such that $e_i = \{v_{i-1}, v_i\} \in E$ for each $i$ $(1 \leq i \leq \ell)$.
 Note that each vertex or edge may appear more than once in a walk. 
 A walk is {\em closed} if $v_\ell = v_0$.
A graph $G$ is {\em connected} 
 if there exists a walk from $u$ to $v$ for any pair of vertices $u, v \in V$.

 For a closed walk $w=v_0 e_1 \cdots e_{\ell} v_{\ell}$, 
  the {\em visit number} of $v \in V$, denoted by $\visit(v)$, 
  is the number that $v$ appears in the walk except for $v_0$ (since $v_0=v_\ell$). 
 A {\em parity Hamiltonian cycle} ({\em PHC} for short) of a graph $G$ is a closed walk
  in which $\visit(v) \equiv 1 \pmod{2}$ holds for each $v \in V$. 
 Remark again that an edge may appear more than once in a PHC $w$.
 Clearly, a graph must be connected to have a PHC, and this paper is basically concerned with connected graphs.

 An {\em edge count vector} $\Vec{x} \in \mathbb{Z}_{\geq 0}^E$ of a closed walk $w$ 
  is an integer vector where $x_e$ for $e \in E$ counts the number of occurrence of $e$ in $w$. 
 Remark that 
\begin{eqnarray}
 \visit(v) = \frac{1}{2}\sum_{e \in \delta(v)} x_e 
\label{eq:visit}
\end{eqnarray}
  holds for any closed walk. 
 Thus, we see that a PHC is a closed walk whose edge count vector $\Vec{x} \in \mathbb{Z}_{\geq 0}^E$ satisfies 
  the {\em parity condition} 
\begin{eqnarray}
 \sum_{e \in \delta(v)} x_e 
 \equiv 2 \pmod{4}
\label{eq:PHC}
\end{eqnarray}
  for each $v \in V$. 

\subsubsection{PHC as an Eulerian cycle of a multigraph}
 As given an arbitrary integer vector $\Vec{x} \in \mathbb{Z}_{\geq 0}^E$, 
  the parity condition \eqref{eq:PHC} is a necessary condition that $\Vec{x}$ is an edge count vector of a PHC. 
 In fact, 
  the following easy but important observation provides an if-and-only-if condition. 
\begin{prop}\label{prop:PHCcondition}
 Let $G=(V,E)$ be an undirected simple graph and 
 let $\Vec{x} \in \mathbb{Z}_{\geq 0}^E$ be an {\em arbitrary} integer vector. 
 Let $F = \{ e \in E \mid x_e > 0\}$, then,   
  $\Vec{x}$ is an edge count vector of a PHC in $G$ 
 if and only if $\Vec{x}$ satisfies \eqref{eq:PHC} and the subgraph $H=(V,F)$ of $G$ is connected. 
\end{prop}

 As a preliminary of the proof of Proposition~\ref{prop:PHCcondition}, 
  we introduce an {\em Eulerian cycle} of a {\em multigraph}. 
 For a simple graph $G=(V,E)$ and 
  any nonnegative integer vector $\Vec{x} \in \mathbb{Z}_{\geq 0}^E$, 
  let ${\cal E}_{\Vec{x}}$ be a {\em multiset} 
   such that $e \in E$ appears $x_e$ times in ${\cal E}_{\Vec{x}}$. 
 Then, let $(G,\Vec{x})$ represent a {\em multigraph} with a vertex set $V$ and a multiedge set ${\cal E}_{\Vec{x}}$. 
 Note that $(G,\Vec{1})=G$ where $\Vec{1} \in \mathbb{Z}_{\geq 0}^E$ denotes the all one vector. 
 We say $(G,\Vec{x})$ is connected if a simple graph $H=(V,F)$ is connected where $F=\{e \in E \mid x_e >0\}$. 
 An {\em Eulerian cycle} of $(G,\Vec{x})$ is a closed walk 
  which uses each element of the multiset ${\cal E}_{\Vec{x}}$ exactly once. 
 It is celebrated fact due to Euler~\cite{Euler} that 
 $(G,\Vec{x})$ has an Eulerian cycle if and only if 
 $(G,\Vec{x})$ is connected and $\Vec{x}$ satisfies the {\em Eulerian condition} 
\begin{eqnarray}
 \sum_{e \in \delta(v)} x_e 
 \equiv 0 \pmod{2}
\label{eq:Euler}
\end{eqnarray}
 holds for any $v \in V$. 

\begin{proof}[Proof of Proposition~\ref{prop:PHCcondition}]
 The `only if' part is easy from the definition. 
 We prove the `if' part. 
 Note that $\Vec{x}$ satisfies \eqref{eq:Euler} since $\Vec{x}$ satisfies \eqref{eq:PHC}. 
 Since $H$ is connected by the hypothesis, the multigraph $(G,\Vec{x})$ has an Eulerian cycle $w$. 
 Considering \eqref{eq:visit}, 
 it is not hard to see that $w$ is a PHC 
  since $\Vec{x}$ satisfies~\eqref{eq:PHC}. 
\end{proof}

 For convenience, 
  we say $\Vec{x} \in \mathbb{Z}_{\geq 0}^E$ {\em admits a PHC} in $G$
  if $\Vec{x}$ is an edge count vector of a PHC in~$G$. 
 In summery, Proposition~\ref{prop:PHCcondition} implies the following.  
\begin{corollary}\label{cor:PHCcondition}
 Let $G=(V,E)$ be an undirected simple graph and 
 let $\Vec{x} \in \mathbb{Z}_{\geq 0}^E$ be an {\em arbitrary} integer vector. 
 Then,   
  $\Vec{x}$ admits a PHC in $G$ 
 if and only if 
  $(G,\Vec{x})$ is connected and 
  $\Vec{x}$ satisfies \eqref{eq:PHC}. 
\end{corollary}

\subsubsection{PHC with an edge constraint}
 As we repeatedly remarked, 
  a PHC may use an edge more than once. 
 For convenience, 
  let ${\rm PHC}_{\z}$ for $z\in \mathbb{Z}_{>0}$ denote a PHC using each edge at most $\z$ times.

\subsection{Other graph terminology}\label{sec:pre-graph}
 This subsection introduces some other graph terminology which we will use in this paper. 
\subsubsection{Fundamental notations}
 A {\em simple path} is a walk $w=v_0e_1v_1e_2 \cdots e_{\ell}v_{\ell}$ 
  which visits every vertex (and hence every edge) at most once, 
  where $\ell \geq 0$ is the {\em length of the path $w$}. 
 Similarly, 
  a {\em simple cycle} is a closed walk $w=v_0e_1v_1e_2 \cdots e_{\ell}v_0$ 
  which visits every vertex at most once, 
  where $\ell \geq 0$ is the {\em length of the cycle~$w$}. 
 An {\em odd cycle} is a simple cycle of odd length. 

 Let $G=(V,E)$ be a graph. 
 For an {\em edge} subset $F \subseteq E$, 
  let $G-F$ denote a graph $H=(V,E \setminus F)$. 
 For a vertex subset $S \subseteq V$, 
  let $G - S$ denote the subgraph induced by $V \setminus S$, i.e., 
  $G - S$ is given by 
   deleting from~$G$ all vertices of $S$ and all edges incident to $S$. 
 For convenience, 
  we simply use $G-e$ for $e \in E$ instead of $G-\{e\}$, and 
  $G-v$ for $v \in V$ as well.  
 For a pair of graphs $G$ and $H$, 
  let $G+H = (V(G) \cup V(H), E(G) \cup E(H))$. 

\subsubsection{$T$-join}
 Let $G=(V,E)$ be a graph and let $T$ be a subset of $V$ such that $|T|$ is even.
 Then, $J \subseteq E$ is a {\em $T$-join}
  if the graph $H = (V,J)$ satisfies
 \begin{equation}
 d_H(v) \equiv \left\{
 \begin{array}{ll}
    1 \pmod{2}\ & \mbox{if } v \in T, \\
    0 \pmod{2}\ & \mbox{otherwise,} 
 \end{array}
 \right. 
 \end{equation}
 for any $v \in V$~\cite{Schrijver}.
 Notice that a graph $H'=(T,J)$ may not be connected, in general. 
 A $T$-join is a generalized notion of a matching,
  meaning that $J$ is a matching when all edges in $J$ are disjoint.
\begin{thm}[cf.~\cite{KV}]\label{thm:T-join}
 For any connected simple graph $G=(V,E)$ and 
  for any $T \subseteq V$ satisfying that $|T|$ is even, 
  $G$ contains a $T$-join. 
\end{thm}
  A $T$-join is found in $O(|V|+|E|)$ time (see Section~\ref{alg-Tjoin}). 

\subsubsection{Edge connectivity}
 A graph is {\em $k$-edge connected} for a positive integer $k$ 
  if the graph remains connected after removing arbitrary $k-1$ edges. 
 A {\em $k$-edge connected component} $H$ of $G$ 
  is a maximal induced subgraph of $G$ such that $H$ is $k$-edge connected. 

 The {\em ear decomposition} is a cerebrated characterization of two-edge connected graphs. 
 An {\em ear} $w = v_0e_1v_1e_2\cdots e_{\ell}v_{\ell}$ of a graph $G$ is 
  a simple path ($v_0=v_{\ell}$ may hold) of length at least one 
  where $v_0$ and $v_\ell$ 
  are in the same two-edge connected component of $G-w$ and 
  $d(v_i)=2$ for each $i=1,\ldots,\ell-1$.  
 A cycle graph, which consists of a simple cycle only, is two-edge connected. 
 It is not difficult to see that 
  any two-edge connected graph, except for a cycle graph, has an ear. 
 By the definition, 
  a graph deleting an ear $w$ except for $v_0$ and $v_{\ell}$ from a two-edge connected graph $G$ is again two-edge connected 
  unless $G$ is a cycle graph. 
 Recursively deleting ears from a two-edge connected graph $G$, we eventually obtain a cycle graph. 
 The sequence of ears in the operation is called {\em ear decomposition} of $G$.
 The following fact is well-known. 
\begin{thm}[cf.~\cite{Schrijver}]
 A graph $G$ is two-edge connected if and only if $G$ has an ear decomposition. 
\end{thm}

\subsubsection{Graph classes}
 Let $P_n$ ($n \geq 2$) denote a graph consisting of a simple path with $n$ vertices. 
 Notice that the {\em length} of the path $P_n$ is $n-1$. 
 Let $C_n$ ($n \geq 3$) denote a cycle graph with $n$ vertices.  
 A graph is {\em $P_k$-free} (resp.\ {\em $C_k$-free}) if it does not contain $P_k$ (resp.\ $C_k$) as an induced subgraph. 
 A graph is {\em $C_{\geq k}$-free}\footnote{
  $C_{\geq k}$-free is often denoted by $C_{n+k}$-free~\cite{BLS87}. 
} if $G$ is $C_{k'}$-free for all $k' \geq k$. 
 Clearly, 
  any $P_k$-free graph is also $P_{k+1}$-free, as well as 
  any $C_{\geq k}$-free graph is also $C_{\geq k+1}$-free.  
 We can also observe that 
  any $P_k$-free graph is $C_{\geq k+1}$-free. 
 However, 
  any $C_{\geq k}$-free is not included in $P_{l}$-free for any $l$, 
  since a tree, clearly $C_{\geq 3}$-free, admits a path of any length. 

 Many important graph classes 
 are known to be characterized as $P_k$-free or  $C_{\geq k}$-free. 
 For instance, 
  {\em cographs} is equivalent to $P_4$-free, 
  {\em chordal} is equivalent to $C_{\geq 4}$-free, and 
  {\em chordal bipartite} is $C_{\geq 6}$-free bipartite (cf.~\cite{BLS87})\footnote{
  Here, we omit the definitions of cograph, chordal and chordal bipartite. 
  This paper requires the properties of $P_k$-free or $C_{\geq k}$-free, only. 
}.

\section{Computational Complexity of The ${\rm PHC}$ Problems}\label{sec:z}
 It may not be trivial if the PHC problem is in {\NP}, 
  since the length of a closed walk is unbounded in the problem. 
 Section~\ref{sec:easy4} completely characterizes the graphs which have PHCs, and 
  shows that the PHC problem is in \PP. 
 Furthermore, 
  if a graph has a PHC then we can find a $\PHC_4$ in linear time. 
 In contrast, Section~\ref{sec:hard123} shows that the ${\rm PHC}_z$ problems is {\NP}-hard for each $z=1,2,3$. 
 Section~\ref{sec:4-edge-connected} further investigates the ${\rm PHC}_3$ problem, and 
  shows that the problem is in \PP\/ for four-edge connected graphs. 
\subsection{The characterization of the graphs which have PHCs}\label{sec:easy4}
 To begin with, we give an if-and-only-if characterization of graphs which have PHCs. 
\begin{thm} \label{thm:suf_and_nec_PHC}
 A connected graph $G=(V,E)$ contains a $\PHC$ 
  if and only if the order $|V|$ is even or $G$ is non-bipartite. 
\end{thm}

 \begin{proof}
 We show the `if' part in a constructive way.  
 First, we are concerned with the case that $|V|$ is even. 
 Let $J \subseteq E$ be a $V$-join of $G$, 
  which always exists by Theorem~\ref{thm:T-join}. 
 Let $\Vec{x} \in \mathbb{Z}_{\geq 0}^E$ be given by 
 \begin{equation}
 x_e = \left\{
 \begin{array}{ll}
 2\ & \mbox{ if } e \in J, \\
 4\ & \mbox{ otherwise.}
 \end{array}
 \right.
\label{assign1-even}
\end{equation}
 Then, 
  $\Vec{x}$ satisfies the parity condition \eqref{eq:PHC}, 
  that is $\sum_{e \in \delta(v)} x_e \equiv 2 \pmod{4}$, for each vertex $v \in V$ 
  since $J$ is a $V$-join of $G$. 
 Clearly, $(G,\Vec{x})$ is connected since $G$ is connected and $x_e \geq 1$ for any $e\in E$. 
 By Corollary~\ref{cor:PHCcondition}, 
  $\Vec{x}$ admits a PHC in $G$. 

\begin{figure}[tbp]
 \begin{center}
\begin{tabular}{c @{\hspace{2em}} c @{\hspace{2em}} c}
	\includegraphics[width=45mm]{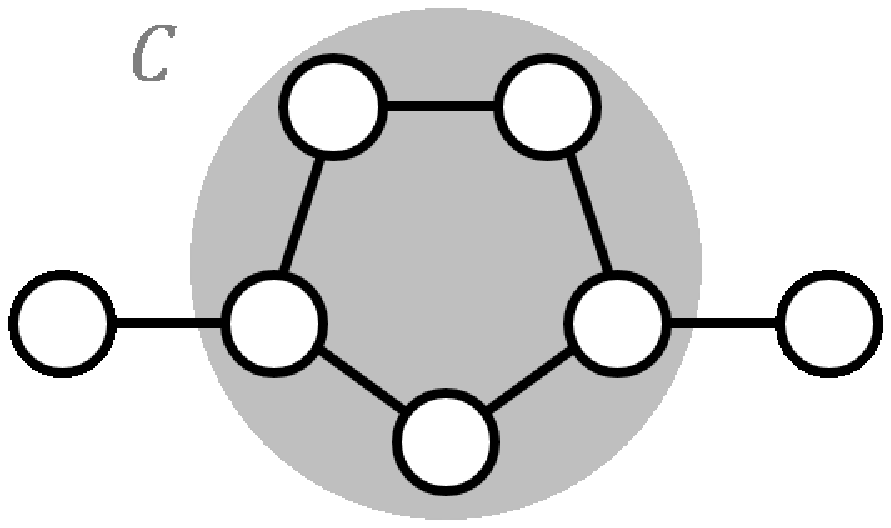} &
	\includegraphics[width=45mm]{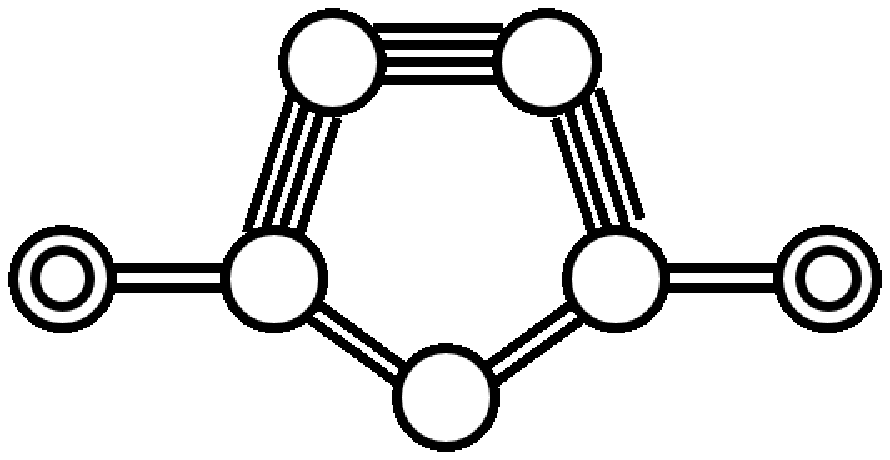} &
	\includegraphics[width=45mm]{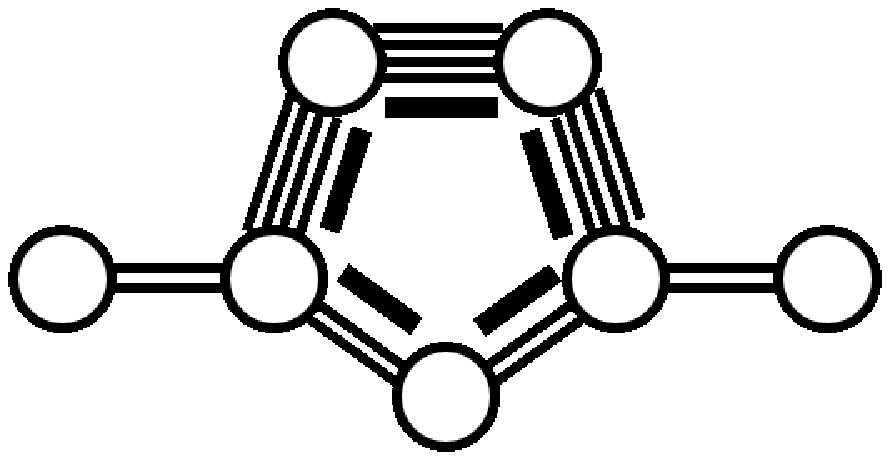} \\
	& & \\
	(a) & (b) & (c)
\end{tabular}
 \end{center}
 \caption{Example for the proof of Theorem~\ref{thm:suf_and_nec_PHC} when $|V|$ is odd and $G$ is non-bipartite. 
  (a) Given graph $G$, and an odd cycle $C$, 
  (b) $\Vec{x}'$ given by \eqref{assign1-tmp1} where double circles denote $T$, 
  (c) $\Vec{x}''$ given by~\eqref{assign1-tmp2}. 
 \label{fig:proof.3.1}}
\end{figure}
 Next, we are concerned with the case that $|V|$ is odd and $G$ is non-bipartite.
 As a preliminary step, 
  let $C$ be an arbitrary odd cycle of $G$. 
 Let $T = V \setminus V(C)$, where notice that $|T|$ is even.
 Using Theorem~\ref{thm:T-join}, 
  let $J \subseteq E$ be a $T$-join. 
 Let $\Vec{x}' \in \mathbb{Z}_{\geq 0}^E$ be given by 
 \begin{equation}
 x'_e = \left\{
 \begin{array}{ll}
 2\ & \mbox{ if } e \in J, \\
 4\ & \mbox{ otherwise.} 
 \end{array}
 \right.
 \label{assign1-tmp1}
 \end{equation}
 Then, $\Vec{x}'$ satisfies the parity condition exactly for each $v \in T$, i.e., 
\begin{eqnarray*}
  \sum_{e \in \delta(v)} x'_e \equiv \left\{
\begin{array}{ll}
 2 \pmod{4} & \mbox{ if } v \in T, \\
 0 \pmod{4} & \mbox{ otherwise, }
\end{array}\right.
\end{eqnarray*}
  hold since $J$ is a $T$-join of $G$.  
 Recalling that $V \setminus T = V(C)$, 
  we modify  $\Vec{x}'$ to $\Vec{x}''$ by adding $E(C)$, i.e., 
 \begin{equation}
 x''_e = \left\{
 \begin{array}{ll}
 x'_e+1 \ & \mbox{ if } e \in E(C), \\ 
 x'_e \ & \mbox{ otherwise. } 
 \end{array}
 \right.
 \label{assign1-tmp2}
 \end{equation}
 This modification increases the visit number of each vertex of $V \setminus T$ exactly by one, and hence  
  $\Vec{x}''$ satisfies the parity condition \eqref{eq:PHC} for each $v \in V$. 
 Clearly, the multigraph $(G,\Vec{x}'')$ is connected, and 
  $\Vec{x}''$ admits a PHC by Corollary~\ref{cor:PHCcondition}.  

 Finally, we show the `only if' part. 
 Suppose that $G=(U,V;E)$ is a bipartite graph with an odd order. 
 Without loss of generality, we may assume that $|U|$ is odd, and hence $|V|$ is even. 
 Assume for a contradiction that $G$ has a PHC. 
 Since $\visit(v)$ of a PHC is odd for each $v \in U \cup V$, 
  $\sum_{v \in U} {\rm visit}(v)$ is odd and 
  $\sum_{v \in V} {\rm visit}(v)$ is even, respectively.  
 On the other hand, 
 any closed walk of $G$ satisfies that
 $\sum_{u \in U} {\rm visit}(u) = \sum_{v \in V} {\rm visit}(v)$ 
 since $G$ is bipartite.  
 Contradiction. 
 \end{proof}

\subsubsection{The $\PHC_4$ problem}\label{sec:phc4}
 The proof of Theorem~\ref{thm:suf_and_nec_PHC} implies a PHC using every edge at most {\em five} times: 
  if $G$ is even order then $\Vec{x}$ given by \eqref{assign1-even} uses every edge at most four times, while 
  if $G$ is odd order and nonbipartite 
   then $\Vec{x}''$ given by  \eqref{assign1-tmp1} and \eqref{assign1-tmp2}
   uses every edge at most five times. 
 An enhancement of Theorem~\ref{thm:suf_and_nec_PHC} is given as follows.  
\begin{thm}\label{thm:suf_and_nec_PHC4}
 A connected graph $G=(V,E)$ contains a $\PHC_{4}$ 
  if and only if the order $|V|$ is even or $G$ is non-bipartite. 
\end{thm}

\begin{proof}
 Notice that \eqref{assign1-tmp1} and \eqref{assign1-tmp2}
  imply that $x''_e \in \{2,3,4,5\}$. 
 Then, we modify $\Vec{x}''$ to $\Vec{x}''' \in \{1,2,3,4\}^E$ by setting
 \begin{eqnarray}
 x'''_e = 1
 \label{assign1-tmp4}
 \end{eqnarray}
  for each $e \in E$ satisfying $x''_e = 5$. 
 This modification preserves 
  the parity condition \eqref{eq:PHC} 
  for each $v \in V$. 
 Clearly, $x_e \geq 1$ for each $e \in E$, meaning that $(G,\Vec{x})$ is connected. 
 By Corollary~\ref{cor:PHCcondition}, we obtain the claim. 
\end{proof}

 We remark that the construction in the proofs of 
  Theorems~\ref{thm:suf_and_nec_PHC} and~\ref{thm:suf_and_nec_PHC4} 
   implies the following fact, 
  since a $T$-join is found in linear time by Theorem~\ref{thm:T-join}.  
\begin{corollary} \label{thm:PHC4_P}
  The ${\rm PHC}_z$ problem for any $z \geq 4$ is solved in linear time. 
\qed
\end{corollary}

\subsubsection{A related topic: a maximization version of the PHC problem}\label{sec:Brigham}
 For a graph which does not have a PHC, 
  a reader may be interested in 
  a maximization version of the problem. 
 The following theorem answers the question. 
 \begin{thm}[cf.~\cite{Brigham}] \label{lem:can_visit_V-1}
 Every connected graph has a closed walk visiting all vertices 
  such that 
  the walk visits at least $|V|-1$ vertices an odd number of times and  
  uses each edge at most four times. 
 \end{thm}

 Theorem~\ref{lem:can_visit_V-1} is suggested by Brigham et al.~\cite{Brigham}, 
  in which they gave a linear time algorithm based on the depth first search.
 We here show a slightly generalized claim using a $T$-join, in an approach different from~\cite{Brigham}. 
 Given a graph $G=(V,E)$ and $S \subseteq V$, 
  an {\em $S$-odd walk} is a closed walk of $G$ which 
  visits every vertex of $S$ an odd number of times and 
  visits every other vertex an even number of times. 
 Clearly, a $V$-odd walk is a PHC of $G$. 
\begin{thm} \label{crl:suf_and_nec_arbitPCW}
 For any graph $G$ and any $S \subseteq V(G)$, 
  $G$ contains an $S$-odd walk 
  if and only if
  $|S|$ is even {\em or}
  $G$ is non-bipartite. 
 Furthermore, we can find an $S$-odd walk which uses each edge at most {\em four} times.
\end{thm}
 Theorem~\ref{crl:suf_and_nec_arbitPCW} (and Theorem~\ref{thm:suf_and_nec_PHC4}) implies Theorem~\ref{lem:can_visit_V-1}, 
   by arbitrarily letting $S \subset V$ satisfy $|S|=|V|-1$ 
   for a bipartite graph $G=(V,E)$ with an odd order. 
 We will use Theorem~\ref{crl:suf_and_nec_arbitPCW} 
  in the proof of Lemma~\ref{prop:bipartite_feasibility} in Section~\ref{sec:all-round}. 

\begin{proof}[Proof of Theorem~\ref{crl:suf_and_nec_arbitPCW}]
 The `only-if' part is (essentially) the same as Theorem~\ref{thm:suf_and_nec_PHC}. 
 The `if' part is also similar to Theorems~\ref{thm:suf_and_nec_PHC} and~\ref{thm:suf_and_nec_PHC4}, as follows. 
 When $|S|$ is even, 
  let $J$ be a $S$-join of~$G$. 
 Let $\Vec{x} \in \mathbb{Z}_{\geq 0}^E$ be given by 
 \begin{equation*}
 x_e = \left\{
 \begin{array}{ll}
 2\ & \mbox{ if } e \in J \\
 4\ & \mbox{ otherwise, }
 \end{array}
 \right.
 \end{equation*}
 then $(G,\Vec{x})$ has an Eulerian cycle, 
  which in fact an $S$-odd walk since $J$ is a $S$-join of $G$.  
 
 When $|S|$ is odd and $G$ is non-bipartite, 
  let $C$ be an odd cycle of $G$. 
 Let 
\begin{eqnarray*}
 T = \{ v \in V(G) \setminus V(C) \mid v \in S \} \cup \{ v \in V(C) \mid v \not\in S \},
\end{eqnarray*}
  and let $J'$ be a $T$-join of $G$. 
 Let $\Vec{x}' \in \mathbb{Z}_{\geq 0}^E$ be given by 
 \begin{equation*}
 x'_e = \left\{
 \begin{array}{ll}
 2\ & \mbox{ if } e \in J \mbox{ and } e \notin E(C), \\
 3\ & \mbox{ if } e \in J \mbox{ and } e \in E(C), \\
 4\ & \mbox{ if } e \notin J \mbox{ and } e \notin E(C), \\
 1\ & \mbox{ if } e \notin J \mbox{ and } e \in E(C), 
 \end{array}
 \right.
 \end{equation*}
 then $\visit(v)$ is odd for each vertex in $S$ since $J'$ is a $T$-join, 
 while $\visit(v)$ is even for others. 
\end{proof}

 \subsection{The ${\rm PHC}_z$ problems when $z=1,2,3$}\label{sec:hard123}
 In Section~\ref{sec:easy4}, 
   we have established that the $\PHC_z$ problem for any $z \geq 4$ is polynomial time solvable. 
 This section shows that 
  the $\PHC_z$ problem is {\NP}-complete for each $z \in \{1,2,3\}$. 
 Remark that the following 
  Theorems~\ref{thm:PHC1_NPC}, \ref{thm:PHC2_NPC}, and \ref{thm:PHC3_NPC} are independent, 
  e.g., the fact that $\PHC_3$ is {\NP}-complete does not imply the fact that $\PHC_2$ is {\NP}-complete, and vice versa.  
 \begin{thm}  \label{thm:PHC1_NPC}
 The ${\rm PHC}_1$ problem is {\NP}-complete.
 \end{thm}
 \begin{proof}
 It is known that
  the HC problem is {\NP}-complete even
  when a given graph is three-edge connected planar cubic~\cite{Garey}.
 It is not difficult to see that
   the PHC problem with $\z=1$ for a cubic graph is exactly same as the HC problem.
 \end{proof}

 \begin{thm}  \label{thm:PHC2_NPC}
 The $\PHC_2$ problem is {\NP}-complete, even when a given graph is two-edge connected. 
 \end{thm}
 \begin{proof}
 We reduce the HC problem for cubic graphs to the $\PHC_2$ problem. 
 Let $G$ be a cubic graph, which is an input of the HC problem.
 Then we construct a graph $H$ as an input of the $\PHC_2$ problem, as follows 
 (see also Figure~\ref{fig:k2-gadget}):
\begin{itemize}
 \item 
  Subdivide every edge $e=\{v,u\} \in E(G)$ into a path of length three, i.e., 
  remove~$e$ and add vertices $v_e$, $u_e$ and edges $\{v,v_e\}$, $\{v_e,u_e\}$, $\{u_e,u\}$. 
 \item For each vertex $v \in V(G)$, 
 attach a cycle of length four, 
 i.e., add vertices $w_{v1}$, $w_{v2}$, $w_{v3}$ and 
  edges $\{v,w_{v1}\}$, $\{w_{v1},w_{v2}\}$, $\{w_{v2},w_{v3}\}$, $\{w_{v3},v\}$. 
 \end{itemize}
 For convenience, 
  let $V$ denote the set of original vertices, i.e., $V=V(G)$, 
  let $V_{\rm s}$ denote the set of vertices $u_e$, $v_e$ added by subdivision, 
    i.e., $|V_{\rm s}|=2|E(G)|$, and 
  let $V_{\rm c}$ denote the set of vertices $w_{v1}$, $w_{v2}$, $w_{v3}$ in attached cycles, 
    i.e., $|V_{\rm s}|=3|V(G)|$, and hence  
   $V(H) = V \cup V_{\rm s} \cup V_{\rm c}$. 
 Then, we show that
 $G$ has a HC if and only if $H$ has a ${\rm PHC}_2$.

 If $G$ has a HC, we claim that a ${\rm PHC}_2$ is in $H$. 
 Suppose that $C \subset E(G)$ is a HC of~$G$. 
 For a path $v\{v,v_e\}v_e\{v_e,u_e\}u_e\{u_e,u\}u$, 
   set $x_{\{v,v_e\}} = x_{\{v_e,u_e\}} = x_{\{u_e,u\}} = 1$ if $e \in C$, 
   otherwise set $x_{\{v,v_e\}} = x_{\{u_e,u\}} = 2$ and $x_{\{v_e,u_e\}} = 0$. 
 For a cycle attached to $v \in V(G)$, 
   set $x_{\{v,w_{v1}\}} = x_{\{w_{v1},w_{v2}\}} = x_{\{w_{v2},w_{v3}\}} = x_{\{w_{v3},v\}} = 1$ 
  (see Figure~\ref{fig:k2-gadget}, right). 
 It is not difficult to see that $\Vec{x}$ indicates a connected closed walk in~$H$, 
  since  
   $\sum_{e' \in \delta_H(v')} x_{e'}$ is even for each $v' \in V(H)$, 
   and HC $C$ is connected in $G$. 
 It is also not difficult to see that 
  every vertex is visited an odd number of times;
  the visit number is three for each vertex in $V$ 
  and one for each vertex in $V_{\rm s} \cup V_{\rm c}$. 
 Hence $\Vec{x}$ admits a $\PHC_2$ of $H$. 

 For the converse, 
  assuming that $H$ has a $\PHC_2$, 
  we claim that $G$ has a HC. 
 Let $\Vec{x}$ be the edge count vector of the $\PHC_2$ of $H$.
 Notice that $d_H(v') = 2$ for $v' \in V_{\rm s} \cup V_{\rm c}$, and 
  it implies that any $\PHC_2$ in $H$ visits every vertex of $V_{\rm s} \cup V_{\rm c}$ exactly once  
 since any $\PHC_2$ is allowed to use every edge at most twice. 
 Then it is not difficult to observe that
 $x_e=1$ for every edge of the attached cycles, that is,
 $x_{\{v,w1\}} = x_{\{w1,w2\}} = x_{\{w2,w3\}} = x_{\{w3,v\}} = 1$.
 Furthermore, there are three possible assignments of 
   $(x_{\{v,v_e\}}, x_{\{v_e,u_e\}}, x_{\{u_e,u\}})$, 
  that is $(1,1,1)$, $(2,0,2)$ or $(0,2,0)$, 
  where $(0,2,0)$ is inadequate because a $\PHC_2$ must be connected. 
 Now, noting that $d_H(v)=5$ holds for each $v \in V$ since $G$ is cubic, 
 let $a,b,c,\{v,w_{v1}\},\{v,w_{v3}\}$ be the edges incident to $v$.
 Then, any assignments
  $x_a$, $x_b$, $x_c$, $x_{\{v,w_{v1}\}}$, $x_{\{v,w_{v3}\}}$ of a $\PHC_2$ 
  must satisfy 
 \begin{equation*}
 x_a + x_b + x_c + x_{\{v,w_{v1}\}} + x_{\{v,w_{v3}\}} \equiv 2 \pmod{4}
 \end{equation*}
  by the parity condition on $v$.  
 As we saw, $x_{\{v,w_{v1}\}} = x_{\{v,w_{v3}\}}=1$ holds, which implies
 \begin{equation*}
 x_a + x_b + x_c \equiv 0 \pmod{4}.
 \end{equation*}
 Since each value of $x_a$, $x_b$, $x_c$ is at most two
  and none of them is equal to zero, 
  exactly two of $x_a$, $x_b$, $x_c$ must be assigned to one and 
  the other must be assigned to two; 
   exactly two of the three subdivided paths incident to $v$
   are assigned to $(1,1,1)$'s
  and the remaining one is assigned to $(2,0,2)$. 
 Now, 
  let $C' \subset E(G)$ be 
   the set of edges corresponding to $(1,1,1)$ paths in a $\PHC_2$ in $H$. 
 From the connectedness of $\PHC_2$ in $H$ it is clear that $C'$ is a HC in $G$.  
 \end{proof}
\begin{figure}[tbp]
 \begin{center}
	\includegraphics[width=105mm]{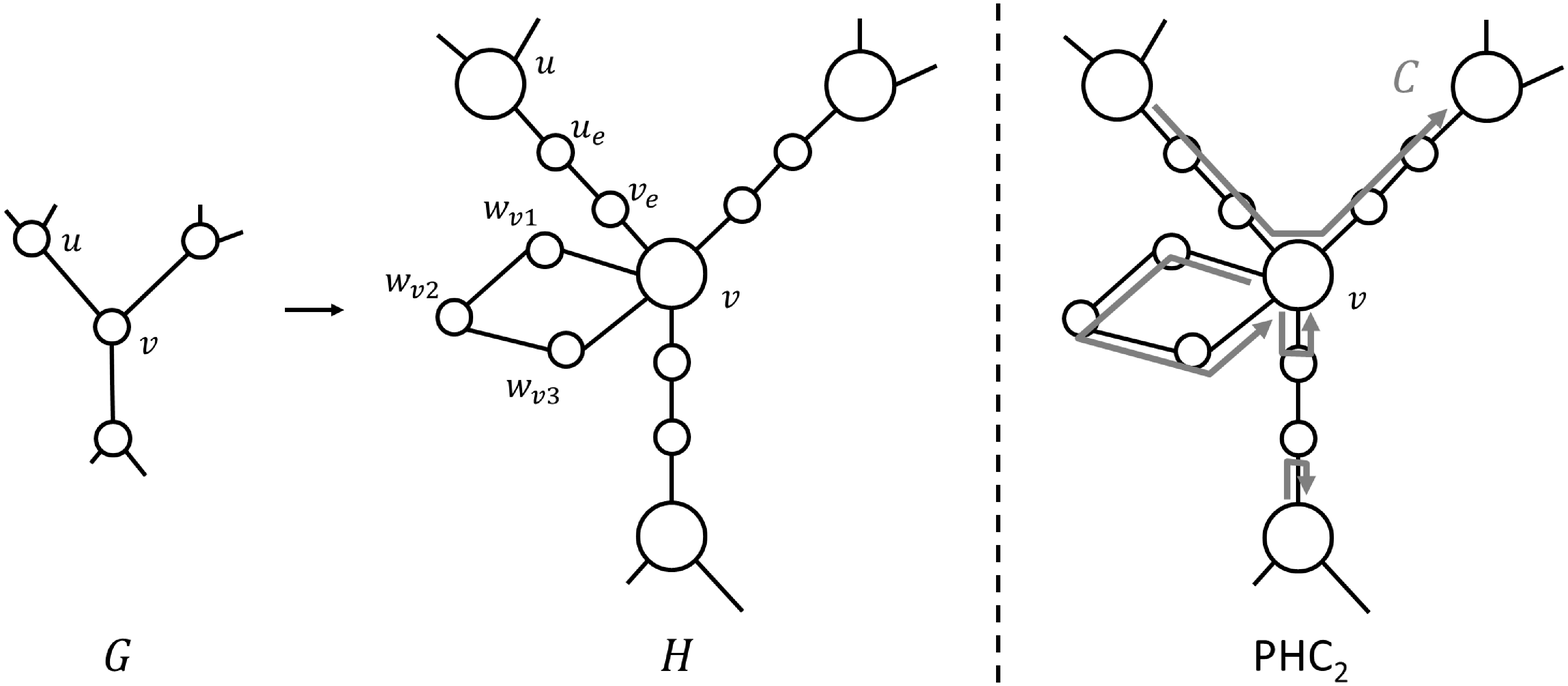}
 \end{center}
 \caption{A ${\rm PHC}_2$ around vertex $v$. \label{fig:k2-gadget}}
\end{figure}
 In a similar way to Theorem \ref{thm:PHC2_NPC}, but more complicated,
  we can show the hardness of the $\PHC_3$ problem.  
\begin{thm} \label{thm:PHC3_NPC}
 The $\PHC_3$ problem is {\NP}-complete, even when a given graph is two-edge connected. 
\end{thm}
\begin{proof}
 We reduce an instance $G$ of the HC problem for cubic graphs to 
  an instance $H$ of the $\PHC_3$ problem 
  exactly in the same way as the proof of Theorem \ref{thm:PHC2_NPC}  
 (see Fig.~\ref{fig:k2-gadget}).
 We here reuse the notations in the proof of Theorem \ref{thm:PHC2_NPC}. 
 If $G$ has an HC, 
   we obtain a ${\rm PHC}_3$ of $H$ by Theorem \ref{thm:PHC2_NPC}, 
  where we remark that a ${\rm PHC}_2$ is a ${\rm PHC}_3$.

 The opposite direction is similar to Theorem~\ref{thm:PHC2_NPC}. 
 Let $\Vec{x}$ be the edge count vector of the $\PHC_3$ of $H$.
 Since $\PHC_3$ visits the vertices in $V_{\rm c}$ an odd number of times,
  the assignment of the edges $\{v,w_{v1}\}, \{w_{v1},w_{v2}\}, \{w_{v2},w_{v3}\}, \{w_{v3},v\}$ are valued to either all ones or all threes. 
 Meanwhile, the possible assignments of $ (x_{\{v,v_e\}}, x_{\{v_e,u_e\}}, x_{\{u_e,u\}})$ are
 \begin{equation} \label{eq:subdivided_path_condition}
 (1,1,1), 
 (2,0,2), 
 (3,3,3), 
 (0,2,0),
 \end{equation}
 where $(0,2,0)$ is inadequate since any $\PHC_3$ is connected.

 Let $a, b, c \in E(H)$ denote three edges incident to $v$ other than $\{v, w_{v1}\}$ and $\{v, w_{v3}\}$.
 Then, any assignments $x_a, x_b, x_c, x_{\{v, w_{v1}\}}, x_{\{v, w_{v3}\}}$ must satisfy
 \begin{equation*}
 x_a + x_b + x_c + x_{\{v, w_{v1}\}} + x_{\{v, w_{v3}\}} \equiv 2 \pmod{4}
 \end{equation*}
 by the parity condition of $v$.
 As we saw, $x_{\{v, w_{v1}\}} = x_{\{v, w_{v3}\}} = 1$ or $3$, both of which implies
 \begin{equation*}
 x_a + x_b + x_c \equiv 0 \pmod{4}.
 \end{equation*}
 Since each of $x_a, x_b, x_c$ is at most three
  and none of them are equal to zero by \eqref{eq:subdivided_path_condition},
  the possible assignments are either
 \begin{equation*} 
 (x_a, x_b, x_c) =  (1, 1, 2) \mbox{ or } (3, 3, 2).
 \end{equation*}
 In the former assignment,
  the assignment for the three subdivided paths incident to $v$ consists of two $(1,1,1)$s and one $(2,0,2)$,
  and two $(3,3,3)$s and one $(2,0,2)$ in the latter case.
 
 Let $C^\prime \in E(G)$ be the set of edges which are corresponding to paths in $H$
  assigned as $(1,1,1)$ and $(3,3,3)$.
 Since $\PHC_3$ is connected,
  $C^\prime$ clearly suggests an HC of $G$. 
 \end{proof}

\subsection{The $\PHC_3$ problem for four-edge connected graphs}\label{sec:4-edge-connected}
 The $\PHC_3$ problem is \NP-complete for {\em two}-edge connected graphs, as we have shown in Theorem~\ref{thm:PHC3_NPC}. 
 This subsection establishes the following. 
 \begin{thm}  \label{thm:PHC3-4EC}
 A {\em four}-edge connected graph $G=(V,E)$ contains a $\PHC_3$
  if and only if the order $|V|$ is even or $G$ is non-bipartite. 
 \end{thm}
 To prove Theorem~\ref{thm:PHC3-4EC}, 
  we use the following celebrated theorem. 
 \begin{thm}[\cite{Nash-Williams,Gusfield}]  \label{lem:4-edge-connected-spanning-tree} 
 Every four-edge connected graph has two edge disjoint spanning trees. 
 \end{thm}
\begin{proof}[Proof of Theorem~\ref{thm:PHC3-4EC}]
 The `only-if' part follows that of Theorem~\ref{thm:suf_and_nec_PHC}. 
 We show the `if' part, in a constructive way. 
 Suppose that $G$ is four-edge connected. 
 Then, 
  let $\tau$ and $\tau'$ be a pair of edge disjoint spanning trees of $G$, 
  implied by Theorem~\ref{thm:PHC3-4EC}. 
 Intuitively, 
  we construct a {\em connected} closed walk on $\tau$, 
  and control the parity condition using edges in $\tau'$, then 
  we obtain a ${\rm PHC}_3$. 

 Let $\Vec{x} \in \mathbb{Z}_{\geq 0}^{E}$ be given by 
 \begin{equation}
 x_e = \left\{
 \begin{array}{ll}
 2 \ & \mbox{ if } e \in E(\tau), \\ 
 0 \ & \mbox{ otherwise. } 
 \end{array}
 \right.
 \label{eq:160503a}
 \end{equation}
 Then, $(G,\Vec{x})$ is connected, and has an Eulerian cycle, say $w$. 
 Let $S$ be the set of vertices with even degree in $\tau$, 
  i.e., $S$ is the entire set of vertices each of which $w$ visits an even number of times. 
 We also remark that $|V \setminus S|$ is even, by the shake-hands-theorem. 
 In the following, we consider two cases whether $|V|$ is even or odd. 

 If $|V|$ is even, then $|S|$ is even. 
 Let $J \subseteq E(\tau')$ be an $S$-join in the tree $\tau'$. 
 Then, 
  let $\Vec{x}' \in \mathbb{Z}_{\geq 0}^E$ be defined by 
 \begin{equation}
 x'_e = \left\{
 \begin{array}{ll}
 x_e+2 \ & \mbox{ if } e \in J, \\ 
 x_e \ & \mbox{ otherwise. } 
 \label{eq:160503b}
 \end{array}
 \right.
 \end{equation}
 It is easy to see that $\Vec{x}'$ 
  satisfies the parity condition~\eqref{eq:PHC} for each vertex of $V$ 
 since $J$ is a $S$-join. 
 Clearly $(G,\Vec{x}')$ is connected, and $\Vec{x}'$ admits a PHC by Corollary~\ref{cor:PHCcondition}. 
 Notice that $J \subseteq E(\tau')$ is disjoint with $E(\tau)$, 
  meaning that \eqref{eq:160503a} and \eqref{eq:160503b} imply that $x'_e \leq 2$ for each  $e \in E$. 
 We obtain the claim in the case.  
 
 If $|V|$ is odd, then $|S|$ is odd. 
 By the hypothesis, 
  $G$ is non-bipartite and hence $G$ contains an odd cycle, say $C$. 
 Then, 
  let $\Vec{x}'' \in \mathbb{Z}_{\geq 0}^E$ be defined by 
 \begin{equation}
 x''_e = \left\{
 \begin{array}{ll}
 x_e+1 \ & \mbox{ if } e \in E(C), \\ 
 x_e \ & \mbox{ otherwise, } 
 \end{array}
 \right.
 \label{eq:160503c}
 \end{equation}
 i.e., increase the value of $x_e$ for every $e \in E(C)$ by one. 
 We can observe from the construction  
  that $(G,\Vec{x}'')$ again has an Eulerian cycle, say $w'$. 
 Let $S'=S \oplus V(C)$, then 
  $S'$ is the entire set of vertices each of which $w'$ visits an even number of times. 
 Since $|S|$ is odd and $|V(C)|$ is odd, 
  $|S'|$ is even. 
 Let $J \subseteq E(\tau')$ be a $S'$-join of $\tau'$, and  
   let $\Vec{x}''' \in \mathbb{Z}_{\geq 0}^E$ be defined by 
 \begin{equation}
 x'''_e = \left\{
 \begin{array}{ll}
 x''_e+2 \ & \mbox{ if } e \in J, \\ 
 x''_e \ & \mbox{ otherwise. } 
 \end{array}
 \right.
 \label{eq:160503d}
 \end{equation}
 Now, it is easy to see that $\Vec{x}'''$ 
  satisfies the parity condition~\eqref{eq:PHC} for each vertex of $V$ 
 since $J$ is a $S'$-join. 
 Clearly $(G,\Vec{x}''')$ is connected, and $\Vec{x}'''$  admits a PHC. 
 Considering the fact that $J \subseteq E(\tau')$ and $E(\tau)$ are disjoint, 
   we observe from \eqref{eq:160503a}, \eqref{eq:160503c} and \eqref{eq:160503d} that 
 \begin{equation}
 x'''_e = \left\{
 \begin{array}{ll}
 3 \ & \mbox{ if } e \in (E(\tau) \cup J) \cap E(C), \\ 
 2 \ & \mbox{ if } e \in (E(\tau) \cup J) \setminus E(C), \\ 
 1 \ & \mbox{ if } e \in E(C) \setminus (E(\tau) \cup J), \\ 
 0 \ & \mbox{ otherwise, } 
 \end{array}
 \right.
 \end{equation}
 hold, meaning that $x'''_e \leq 3$ for each $e \in E$.  
 We obtain the claim. 
\end{proof}
 
 It is known (cf.~\cite{Schrijver}) that 
  edge disjoint spanning trees $\tau$ and $\tau'$ in four-edge connected graph
  are found in polynomial time 
  (e.g., $\Order(|E(G)|^2)$ time, due to Roskind, Tarjan~\cite{Roskind-Tarjan}). 
 Thus, the proof of Theorem~\ref{thm:PHC3-4EC} also implies 
  a polynomial time (e.g., $\Order(|E(G)|^2)$) algorithm 
  to find a  $\PHC_3$ in a four-edge connected graph $G$.

\section{The $\PHC_3$ Problem for Two-Edge Connected Graphs} \label{sec:PHC3}
 The $\PHC_3$ problem is \NP-complete for {\em two}-edge connected graphs (Theorem~\ref{thm:PHC3_NPC}), 
 while it is solved in polynomial time for any {\em four}-edge connected graph (Theorem~\ref{thm:PHC3-4EC}). 
 As an approach to {\em three}-edge connected graphs, 
  this section further investigates the $\PHC_3$ problem using the celebrated ear decomposition. 
 This section establishes the following theorem. 
\begin{thm}\label{thm:PHC3-forC5P6free}
 Suppose that a two-edge connected graph $G$ is $P_6$-free or $C_{\geq 5}$-free. 
 Then, $G=(V,E)$ contains a ${\rm PHC}_3$ if and only if the order $|V|$ is even or $G$ is non-bipartite. 
\end{thm}
 Notice that $C_{\geq 5}$-free contains some important graph classes 
  such as chordal (equivalent to $C_{\geq 4}$-free), 
  chordal bipartite (equivalent to $C_{\geq 5}$-free bipartite), 
  and cograph (equivalent to $P_4$-free). 
 We also remark that the Hamiltonian cycle problem is {\NP}-complete 
  for $C_{\geq 4}$-free graphs, as well as $P_5$-free graphs~(cf.~\cite{BLS87}). 

\subsection{Preliminary}\label{sec:all-round}
 For the purpose, we introduce the notion of the {\em all-roundness} of a graph in Section~\ref{sec:all-round}. 
\subsubsection{Generalized problem}
 Section~\ref{sec:PHC3} is actually concerned with the following problem, slightly generalizing the $\PHC_3$ problem. 
\begin{problem}\label{prob:gen}
 Given a graph $G=(V,E)$ and a map $f \colon V \to \{0,1,2,3\}$, 
  find $\Vec{x} \in \{0,1,2,3\}^E$ satisfying the conditions that 
\begin{align} 
 & \sum_{e \in \delta(v)} x_e \equiv f(v) \pmod{4} && \mbox{ for any $v \in V$, } \label{eq:modulo_constraint} \\
 & \mbox{$(G, \Vec{x})$ is connected. } \label{eq:connectivity_constraint} 
\end{align}
\end{problem}
 Clearly, 
  $\PHC_3$ is given by setting 
  $f(v)=2$ for any $v \in V$ 
 (recall Corollary~\ref{cor:PHCcondition}). 
 For convenience, we call $\Vec{x} \in \{0,1,2,3\}^E$ 
  a {\em mod-4 $f$-factor} of $G$ 
  if $\Vec{x}$ satisfies \eqref{eq:modulo_constraint}. 
 A mod-4 $f$-factor $\Vec{x} \in \{0,1,2,3\}^E$ is {\em connected} if it satisfies \eqref{eq:connectivity_constraint},  
  i.e.,  Problem~\ref{prob:gen} is a problem to find a {\em connected mod-4 $f$-factor}. 
 We remark the following two facts. 
\begin{prop}\label{prop:gen-handshake}
 A graph $G=(V,E)$ has a mod-4 $f$-factor 
   only when the map $f$ satisfies that 
\begin{eqnarray}
 \sum_{v \in V} f(v) \mbox{ is even. }
\label{eq:fiseven}
\end{eqnarray} 
\end{prop}
\begin{proof}
 Summing up \eqref{eq:modulo_constraint} over $V$, we obtain 
\begin{equation}
 \sum_{v \in V} f(v) 
 \equiv
 \sum_{v \in V} \sum_{e \in \delta(v)} x_e \pmod{4}.
\label{eq160504a}
\end{equation}
 It is not difficult to see that 
\begin{equation}
 \sum_{v \in V} \sum_{e \in \delta(v)} x_e 
 =
 2\sum_{e \in E} x_e 
\label{eq160504b}
\end{equation}
 holds, in an analogy with the handshaking lemma, that is 
  $\sum_{v \in V} \sum_{e \in \delta(v)}1 = \sum_{e \in E} 2$. 
 By \eqref{eq160504a} and \eqref{eq160504b}, 
  we obtain that 
\begin{equation*}
 \sum_{v \in V} f(v) 
 \equiv
 2\sum_{e \in E} x_e \pmod{4}
\end{equation*}
 which implies the claim. 
\end{proof}

 We will later show in Lemma~\ref{lem:mod-4_f-factor}
  that \eqref{eq:fiseven} is also sufficient for any connected non-bipartite graph 
  to have a mod-4 $f$-factor. 
 For bipartite graphs, 
  we need an extra necessary condition on $f$. 
\begin{prop} \label{prop:bipartite_feasibility}
 A bipartite graph $G=(U,V;E)$ has a mod-4 $f$-factor 
   only when the map $f$ satisfies that 
\begin{equation} \label{eq:bipartite_feasibility}
 \sum_{v\in U} f(v) \equiv \sum_{v\in V} f(v) \pmod{4}. 
\end{equation}
\end{prop}
\begin{proof}
Since $G=(U,V;E)$ is bipartite, 
\begin{equation}
  \sum_{v \in U}\sum_{e \in \delta(v)} x_e
= \sum_{v \in V}\sum_{e \in \delta(v)} x_e 
\label{tmp160504d}
\end{equation}
 is required. 
 Summing up \eqref{eq:modulo_constraint} over $U$ and $V$, respectively, we obtain 
\begin{eqnarray*}
 \sum_{v\in U} f(v) & \equiv & \sum_{v \in U}\sum_{e \in \delta(v)} x_e \pmod{4}, \hspace{1em} \mbox{and}  \\
 \sum_{v\in V} f(v) & \equiv & \sum_{v \in V}\sum_{e \in \delta(v)} x_e \pmod{4}  
\end{eqnarray*}
 hold, which with \eqref{tmp160504d} implies the claim. 
 \end{proof}
 Notice that the condition \eqref{eq:bipartite_feasibility} implies \eqref{eq:fiseven}. 
 We will show in  Lemma~\ref{lem:mod-4_f-factor} 
  that \eqref{eq:bipartite_feasibility} is also sufficient for any connected bipartite graph. 

\subsubsection{All-roundness}
 Then, we introduce the following two notions.  
\begin{dfn}
\label{dfn:all-roundness}
  A graph $G$ is {\em all-round}
   if $G$ has a connected mod-4 $f$-factor for any map $f$ satisfying \eqref{eq:fiseven}.  
\end{dfn}
\begin{dfn}
\label{dfn:bipartite_all-roundness}
 A bipartite graph $G$ is {\em bipartite all-round} 
 if $G$ has a connected mod-4 $f$-factor for any map $f$ satisfying~\eqref{eq:bipartite_feasibility}.
\end{dfn}	
 It is not difficult to see that 
  Definitions~\ref{dfn:all-roundness} and \ref{dfn:bipartite_all-roundness} are 
  (too strong) sufficient condition that a graph contains a $\PHC_3$. 
 In the rest of Section~\ref{sec:PHC3}, we will show the following theorems, which immediately implies Theorem~\ref{thm:PHC3-forC5P6free}.  
\begin{thm}  \label{thm:c5_univ}
 Every two-edge connected $C_{\geq 5}$-free graph is either all-round or bipartite all-round. 
\end{thm}
\begin{thm} \label{thm:p6_univ}
 Every two-edge connected $P_6$-free graph,
  except for $C_5$,
  is either all-round or bipartite all-round. 
\end{thm}

\subsubsection{All-roundness of small graphs}\label{sec:base}
As a preliminary step of the proof, we remark the following facts, which are confirmed by brute force search. 
\begin{prop}  \label{lem:k3_univ}
 $C_3$ (i.e., $K_3$) is all-round (See Figure~\ref{fig:K3}). 
\end{prop}
 \begin{figure}[tbp]
    \begin{center}
        \includegraphics[width=90mm]{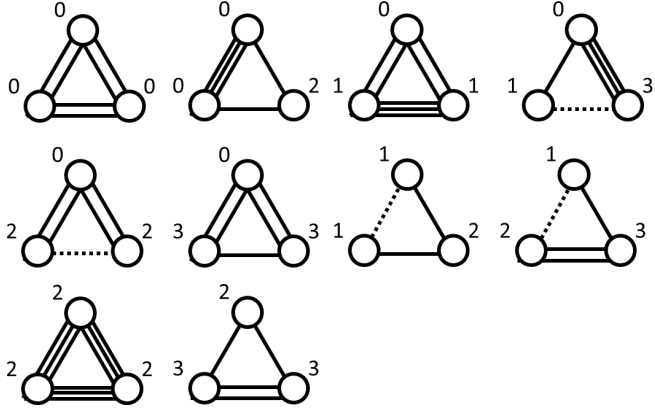}
    \end{center}
        \caption{$K_3$ is all-round. \label{fig:K3}}
 \end{figure}
 
 
 \begin{prop}  \label{lem:c4_c6_univ}
 $C_4$ and $C_6$ are bipartite all-round, respectively.
 \end{prop}
 
 The following fact may be counterintuitive.
 \begin{prop} \label{prop:C5_counterexample}
 $C_5$ is {\em not} all-round.
 \end{prop}
 Figure~\ref{fig:c5_not_univ} shows an example of $f$ 
  for which Problem~\ref{prob:gen} does not have a solution. 
 Notice that 
  $C_5$ clearly has a $\PHC_3$. 
  \begin{figure}[tbp]
    \begin{tabular}{cc}
        \begin{minipage}{0.45\textwidth}
            \begin{center}
                \includegraphics[width=33mm]{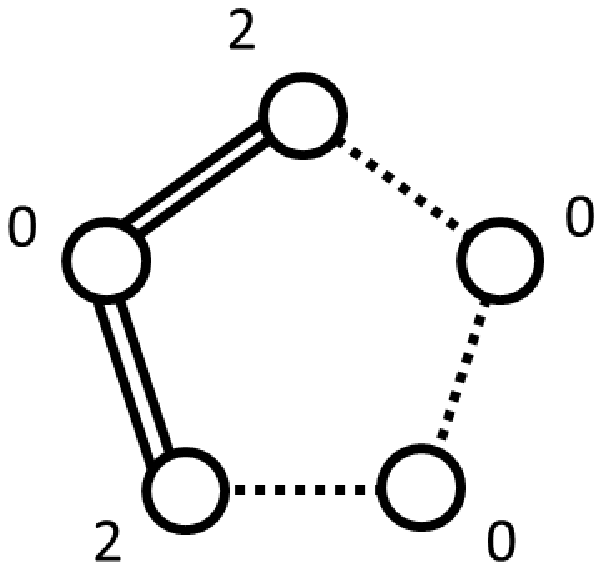}
            \end{center}
        \end{minipage}
        \begin{minipage}{0.45\textwidth}
            \begin{center}
                \includegraphics[width=33mm]{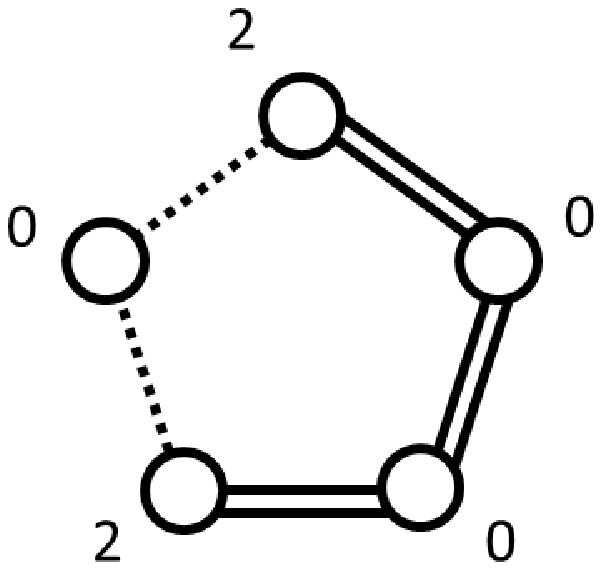}
            \end{center}
        \end{minipage}
    \end{tabular}
    \caption{A counterexample of the all-roundness ($C_5$). \label{fig:c5_not_univ}}
 \end{figure}

\subsection{All-roundness of $C_{\geq 5}$-free: Proof of Theorem~\ref{thm:c5_univ}}\label{sec:c5_univ}
 This section shows Theorem~\ref{thm:c5_univ}, 
  presenting some useful idea on a mod-4 $f$-factor of a graph, 
  and all-roundness or bipartite all-roundness. 
 To begin with, we show for any appropriate map $f$ that 
  it is easy to find a mod-$4$ $f$-factor, which may be disconnected. 
 \begin{lem} \label{lem:mod-4_f-factor}
 Any connected non-bipartite graph has a mod-$4$ $f$-factor for any map $f$ satisfying~\eqref{eq:fiseven}. 
 Any connected bipartite graph has a mod-$4$ $f$-factor for any map $f$ satisfying~\eqref{eq:bipartite_feasibility}. 
 \end{lem}
 \begin{proof}
 We give a constructive proof. 
 Let $T := \{v \in V \mid f(v) \mbox{ is odd} \}$. 
 We remark that $|T|$ is even since $\sum_{v \in V}f(v)$ is even by \eqref{eq:fiseven}. 
 Then, 
  let $J \subseteq E$ be a $T$-join, and  
  let $\Vec{x} \in \{0,1,2,3\}^E$ be given by 
\begin{equation*} 
 x_e = 
 \left\{\begin{array}{ll} 
  1 & \mbox{if } e \in J \\ 
  0 & \mbox{otherwise. } 
  \end{array}\right.
\end{equation*}
 Let $f' \colon V \to \{0,1,2,3\}$ be 
\begin{equation} \label{eq:f_prime}
 f'(v) = \left(f(v) - \sum_{e \in \delta(v)} x_e\right)  \bmod{4}. 
\end{equation}
 Remark that $f'(v)$ is even for any $v \in V$, i.e., $f'(v) = 0 \mbox{ or } 2$ for any $v \in V$. 
 Let $S = \{v \in V \mid f' (v) = 2 \}$. 
 If $|S|$ is even, then 
  let $J'$ be a $S$-join and 
  let $\Vec{x}' \in \{0,1,2,3\}^E$ be defined by 
\begin{equation*} 
 x'_e = 
 \left\{\begin{array}{ll} 
  x_e+2 & \mbox{if } e \in J' \\ 
  x_e & \mbox{otherwise. } 
  \end{array}\right.
\end{equation*}
 It is not difficult to observe 
   that $\Vec{x}'$ satisfies \eqref{eq:modulo_constraint}, and $x'_e \leq 3$ holds for any $e \in E$. 
 Thus, we obtain the claim in the case. 
 Here we remark that if $G$ is bipartite then 
  $|S|$ is even, 
 since \eqref{eq:bipartite_feasibility} implies  
\begin{equation*}
 \begin{array}{ll}
 \displaystyle \sum_{v \in U \cup V} f^\prime (v) &= \displaystyle \sum_{v \in U} f^\prime (v) + \sum_{v \in V} f^\prime (v) \\
 &\equiv \displaystyle \sum_{v \in U} f^\prime (v) - \sum_{v \in V} f^\prime (v) \\
 &\equiv 0 \pmod{4}.
 \end{array}
 \end{equation*}

 If $|S|$ is odd, we need an extra process. 
 Notice that $G$ is non-bipartite in the case, since the above argument. 
 Let $C$ be an odd cycle of $G$ and let $\Vec{x}'' \in \{0,1,2,3\}^{E(G)}$ be 
\begin{equation*} 
 x''_e = 
 \left\{\begin{array}{ll} 
  x_e+1 & \mbox{if } e \in E(C) \\ 
  x_e & \mbox{otherwise. } 
  \end{array}\right.
\end{equation*}
 Let $f'' \colon V \to \{0,1,2,3\}$ be 
\begin{equation} \label{eq:f_dprime}
 f''(v) = \left(f(v) - \sum_{e \in \delta(v)} x''_e\right)  \bmod{4}. 
\end{equation}
 Let 
  $S' = \{v \in V(G) \mid f'' (v) = 2 \}$. 
 Then, $S' = S \oplus V(C)$ holds, 
  which implies $|S'|$ is even  
  since $|S|$ and $|V(C)|$ are respectively odd. 
 Let $J' \subseteq E(G)$ be an $S'$-join and 
 let $\Vec{x}''' \in \{0,1,2,3\}^{E(G)}$ be 
\begin{equation*} 
 x'''_e = 
 \left\{\begin{array}{ll} 
  x''_e+2 \pmod{4} & \mbox{if } e \in J' \\ 
  x''_e & \mbox{otherwise. } 
  \end{array}\right.
\end{equation*}
Then we obtain a mod-4 $f$-factor. 
 \end{proof}

 To obtain a {\em connected} mod-4 $f$-factor, 
  the notions of all-roundness and bipartite all-roundness play an important role. 
 The following lemma is an easy observation from the definition. 
 \begin{lem}[connecting lemma]\label{lem:univ_spanning_subgraph}
 Let $\Vec{x}$ be a mod-4 $f$-factor of $G$, and 
  let $H_1$ and $H_2$ be a distinct pair of connected components of $(G,\Vec{x})$. 
 Suppose that there is a connected subgraph $H$ of $G$ 
   such that 
   $V(H)$ intersects both $V(H_1)$ and $V(H_2)$, 
  and that 
   $H$ is all-round or bipartite all-round. 
 Then, $G$ has another mod-4 $f$-factor $\Vec{x}'$ 
  such that 
   $H_1$ and $H_2$ are connected in $(G,\Vec{x}')$ 
   where other connected components are respectively kept being connected.  
 \end{lem}
 \begin{proof}
 As given $\Vec{x} \in \{0,1,2,3\}^{E(G)}$ and $H$ described in the hypothesis, 
   we define a map $f_H \colon V(H) \to \{0,1,2,3\}$ by 
\begin{equation*}
 f_H(v) = \sum_{e \in \delta_H(v)} x_e \bmod{4}.
\end{equation*}
 Let $\Vec{y} \in \{0,1,2,3\}^{E(H)}$ be defined by $y_e=x_e$ for each $e \in E(H)$, then 
  clearly $\Vec{y}$ is a mod-4 $f_H$-factor of $H$. 
 Proposition~\ref{prop:gen-handshake} implies that 
  $f_H$ satisfies \eqref{eq:fiseven}, as well as 
 Proposition~\ref{prop:bipartite_feasibility} 
  implies that $f_H$ satisfies \eqref{eq:bipartite_feasibility} if $H$ is bipartite. 
 The hypothesis that $H$ is all-round or bipartite all-round implies that 
  there is a {\em connected} mod-4 $f_H$-factor $\Vec{y}' \in \{0,1,2,3\}^{E(H)}$ of $H$. 
Let $\Vec{x}' \in \{0,1,2,3\}^{E(G)}$ be 
\begin{equation*}
 x'_e = \left\{
 \begin{array}{ll}
  y'_e & \mbox{if } e \in E(H), \\
  x_e & \mbox{otherwise,}
 \end{array}
\right.
\end{equation*}
 then, it is not difficult to observe that $\Vec{x}'$ is a desired mod-4 $f$-factor of $G$ 
 (See also Figure~\ref{fig:combining_strategy}).
\end{proof} 
 
 \begin{figure}[tbp]
	\begin{tabular}{cc}
		\begin{minipage}{0.45\textwidth}
			\begin{center}
				\includegraphics[width=40mm]{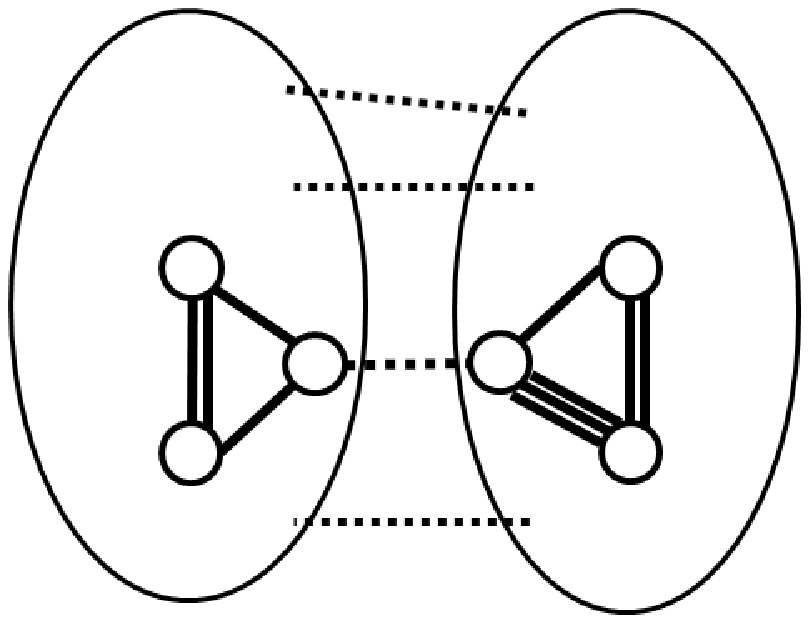}
			\end{center}
		\end{minipage}
		\begin{minipage}{0.45\textwidth}
			\begin{center}
				\includegraphics[width=40mm]{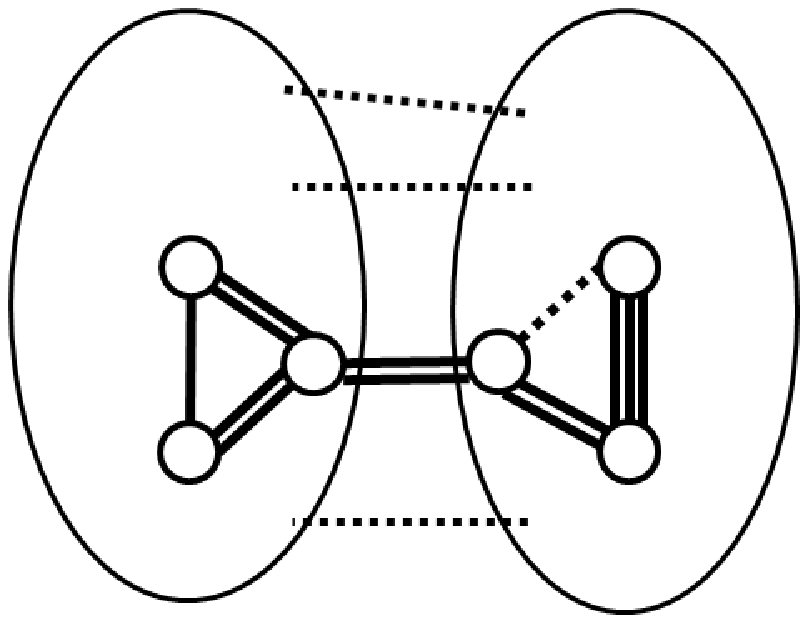}
			\end{center}
		\end{minipage}
	\end{tabular}
	\caption{Connecting lemma. 
    Left figure shows that $H_1$ and $H_2$ are disconnected in $(G,\Vec{x})$. 
    Right figure shows that a mod-4 $f$-factor in $H$ is replaced by another {\em connected} one. 
    The hypothesis that $H$ is all-round or bipartite all-round implies such a connected mod-4 $f$-factor in~$H$. 
\label{fig:combining_strategy}}
\end{figure}
 
 It is not difficult to see that 
  Lemmas~\ref{lem:mod-4_f-factor} and \ref{lem:univ_spanning_subgraph} 
  imply the following useful lemma. 
\begin{lem}  \label{crl:univ_all_edge}
 Let $G$ be a graph. 
 Suppose that 
  for any edge $e \in E(G)$ there exists a subgraph $H$ of $G$ such that $e \in E(H)$ and $H$ is 
  all-round or bipartite all-round. 
 Then, $G$ is all-round, or bipartite all-round. 
\end{lem}
\begin{proof}
 For any appropriate $f$, 
  meaning that $f$ satisfies \eqref{eq:fiseven}, and \eqref{eq:bipartite_feasibility} if $G$ is bipartite, 
  Lemma~\ref{lem:mod-4_f-factor} implies a mod-4 $f$-factor $\Vec{x}$. 
 Since the hypothesis, 
  we can obtain a connected mod-4 $f$-factor by 
  iteratively applying the connecting lemma (Lemma~\ref{lem:univ_spanning_subgraph}) to $\Vec{x}$. 
\end{proof}
  In fact, the all-roundness of $C_{\geq 5}$-free (Theorem~\ref{thm:c5_univ}) is immediate from Lemma~\ref{crl:univ_all_edge}. 
\begin{proof}[Proof of Theorem~\ref{thm:c5_univ}]
 Let $G$ be a two-edge connected $C_{\geq 5}$-free graph. 
 Then, any edge $e$ of $G$ is contained in $C_3$ or $C_4$. 
 Since $C_3$ is all-round (Proposition~\ref{lem:k3_univ}), and 
  $C_4$ is bipartite all-round (Proposition~\ref{lem:c4_c6_univ}), 
 Lemma~\ref{crl:univ_all_edge} implies that $G$ is all-round or bipartite all-round. 
\end{proof}
 We will use Lemma~\ref{crl:univ_all_edge} again 
  in the proof of Theorem~\ref{thm:p6_univ},  
  with some extra arguments. 
 We also show another example in Section~\ref{sec:dense} 
  where we apply Lemma~\ref{sec:dense} to dense graphs. 

\subsection{All-roundness of $P_6$-free: Proof of Theorem~\ref{thm:p6_univ}}\label{sec:p6_univ}
 We cannot directly apply Lemma~\ref{crl:univ_all_edge} to a $P_6$-free graph $G$, 
  since $G$ may contain $C_5$ that is NOT all-round (recall Proposition~\ref{prop:C5_counterexample}). 
 To prove Theorem~\ref{thm:p6_univ}, 
  we investigate the all-roundness (or bipartite all-roundness) of two-edge connected graphs, 
  considering ear-decomposition in Section~\ref{sec:ear}.

\subsubsection{Ear-decomposition for mod-$4$ all-round graphs}\label{sec:ear}
 This section presents three lemmas, which claim a short ear preserves the all-roundness or bipartite all-roundness.  
\begin{lem}  \label{lem:less7_univ}
 Let $G$ and $G'$ be two-edge connected non-bipartite graphs, 
  where $G'$ is given by adding to $G$ an ear of length at most seven. 
 If $G$ is all-round, then $G'$ is all-round. 
\end{lem}
\begin{proof}
 The proof idea is similar to Lemma~\ref{crl:univ_all_edge}: 
 construct a mod-4 $f$-factor which may not be connected (Lemma~\ref{lem:mod-4_f-factor}), and 
 let it be connected using the connecting lemma (Lemma~\ref{lem:univ_spanning_subgraph}) 
  on assumption that $G$ is all-round. 
 A technical issue of the proof idea is that the ear is not bipartite all-round. 
 We show that a desired mod-4 $f$-factor always exist if the ear is short.  

 Let $p=v_0e_1v_1e_2\cdots e_{\ell}v_{\ell}$ denote the ear added to $G$, and 
 let $P=(\{v_0,v_1,\ldots,v_{\ell}\},\{e_1,\ldots,e_{\ell}\})$, where $\ell \leq 7$. 
 Remark that $v_0,v_{\ell} \in V(G)$. 
 As given an arbitrary map $f \colon V(G') \to \{0,1,2,3\}$ satisfying \eqref{eq:fiseven}, 
  let $\Vec{x} \in \{0,1,2,3\}^{E(G')}$ be a mod-4 $f$-factor implied by Lemma~\ref{lem:mod-4_f-factor}. 
 For convenience, 
  let $\Vec{x}^{(0)} = (x_1,\ldots,x_{\ell}) \in \{0,1,2,3\}^{E(P)}$ 
  where $x_i$ denotes $x_{e_i}$ for simplicity. 
 We also define $\Vec{x}^{(a)} \in \{0,1,2,3\}^{E(P)}$ for each $a\in \{1,2,3\}$ by 
\begin{eqnarray*}
 x^{(a)}_i &=& 
 \left\{ \begin{array}{ll} 
  (x_i + a) \bmod 4 & \mbox{if $i$ is odd,} \\
  (x_i - a) \bmod 4 & \mbox{if $i$ is even.} 
 \end{array}\right. 
\end{eqnarray*} 
 Notice that $x^{(a)}_i$ for each $a \in \{0,1,2,3\}$ is a mod-4 $f_P^{(a)}$-factor 
 for a map $f_P^{(a)} \colon V(P) \to \{0,1,2,3\}$ given by 
\begin{eqnarray*}
 f_P^{(a)}(v) &=& 
 \left\{ \begin{array}{ll} 
  (x_1+a) \bmod{4} & \mbox{if } v = v_0, \\
  (x_i +x_{i+1}) \bmod{4} = f(v) & \mbox{if } v \in \{v_1,\ldots,v_{\ell-1}\} , \\
  (x_{\ell}+(-1)^{\ell}a)  \bmod{4} & \mbox{if } v = v_{\ell}.  
 \end{array}\right. 
\end{eqnarray*}
 At the same time, let $f_G^{(a)} \colon V(G) \to \{0,1,2,3\}$ be defined for each $a \in \{0,1,2,3\}$ by 
\begin{eqnarray*}
 f_G^{(a)}(v) &=& 
 \left\{ \begin{array}{ll} 
  \left(f(v) - f_P^{(a)}(v)\right) \bmod{4} & \mbox{if } v \in \{v_0,v_{\ell}\}, \\
  f(v) & \mbox{otherwise}, 
 \end{array}\right. 
\end{eqnarray*}
 then $G$ has a {\em connected} mod-4 $f_G^{(a)}$-factor $\Vec{y}^{(a)} \in \{0,1,2,3\}^{E(G)}$ 
  since $G$ is all-round by the hypothesis. 
 If $(P,\Vec{x}^{(a)})$ consists of at most two connected component, 
  i.e., 
  $x^{(a)}_i = 0$ holds at most one $i \in \{1,2,\ldots,\ell\}$, 
  then $\Vec{x}^{(a)}$ and $\Vec{y}^{(a)}$ implies a connected mod-4 $f$-factor of $G'$.

 Then, we claim that there is $a \in \{0,1,2,3\}$ such that 
  $(P,\Vec{x}^{(a)})$ consists of at most two connected component.  
 Remark that $\{ x_i^{(0)}, x_i^{(1)}, x_i^{(2)}, x_i^{(3)} \} = \{0,1,2,3\}$ holds for each $i \in \{1,2,\ldots,\ell\}$. 
 By a version of the pigeon hole principle, 
  the hypothesis of $\ell \leq 7$ implies that  
  there is an index $a \in \{0,1,2,3\}$ such that $x^{(a)}_j = 0$ holds for at most one $j \in \{1,2,\ldots,\ell\}$; 
 otherwise the multiset $\{ x_j^{(a)} \mid a \in \{0,1,2,3\},\ j \in \{1,2,\ldots,\ell\} \}$ contains 8 or more 0's. 
 The $\Vec{x}^{(a)}$ is the desired mod-4 $f_P^{(a)}$-factor for $P$, and 
 we obtain the claim. 
\end{proof}
 
 The following lemma is a version of Lemma~\ref{lem:less7_univ} for bipartite graphs. 
 The proof is essentially the same, and we omit it. 
\begin{lem}  \label{lem:bip_less7_bip_all-round}
 Let $G$ and $G'$ be two-edge connected bipartite graphs, 
  where $G'$ is given by adding to $G$ an ear of length at most seven. 
 If $G$ is bipartite all-round, then $G'$ is bipartite all-round. 
\qed
\end{lem}

 Finally, we show the following lemma, which 
  claims a connection between a bipartite all-round graph and an all-round graph. 
\begin{lem}  \label{lem:bip_less3_all-round}
 Let $G$ be a two-edge connected bipartite graph, and 
 let $G'$ be a two-edge connected non-bipartite graph given by adding to $G$ an ear of length at most three. 
 If $G$ is bipartite all-round, then $G'$ is all-round. 
\end{lem}
\begin{proof}
 The proof is similar to Lemma~\ref{lem:less7_univ}. 
 Let $p=v_0e_1v_1e_2\cdots e_{\ell}v_{\ell}$ denote the ear added to $G$ where $\ell \leq 3$,  and 
 let $P=(\{v_0,v_1,\ldots,v_{\ell}\},\{e_1,\ldots,e_{\ell}\})$. 
 Remark that $v_0,v_{\ell} \in V(G)$. 
 As given an arbitrary map $f \colon V(G') \to \{0,1,2,3\}$ satisfying \eqref{eq:fiseven}, 
  let $\Vec{x} \in \{0,1,2,3\}^{E(G')}$ be a mod-4 $f$-factor implied by Lemma~\ref{lem:mod-4_f-factor}. 
 For convenience, 
  let $\Vec{x}^{(0)} = (x_1,\ldots,x_{\ell}) \in \{0,1,2,3\}^{E(P)}$ 
  where $x_i$ denotes $x_{e_i}$ for simplicity. 
 We also define $\Vec{x}^{(2)} \in \{0,1,2,3\}^{E(P)}$ by 
\begin{eqnarray*}
 x^{(2)}_i = (x_i + 2) \bmod 4. 
\end{eqnarray*} 
 Notice that $x^{(2)}_i$ is also a mod-4 $f_P^{(2)}$-factor 
 for a map $f_P^{(2)} \colon V(P) \to \{0,1,2,3\}$ given by 
\begin{eqnarray*}
 f_P^{(2)}(v) &=& 
 \left\{ \begin{array}{ll} 
  (x_1+2) \bmod{4} & \mbox{if } v = v_0, \\
  (x_i +x_{i+1}) \bmod{4} = f(v) & \mbox{if } v \in \{v_1,\ldots,v_{\ell-1}\} , \\
  (x_{\ell}+2)  \bmod{4} & \mbox{if } v = v_{\ell}.  
 \end{array}\right. 
\end{eqnarray*}
 At the same time, let $f_G^{(a)} \colon V(G) \to \{0,1,2,3\}$ be defined for each $a \in \{0,2\}$ by 
\begin{eqnarray*}
 f_G^{(a)}(v) &=& 
 \left\{ \begin{array}{ll} 
  \left(f(v) - f_P^{(a)}(v)\right) \bmod{4} & \mbox{if } v \in \{v_0,v_{\ell}\}, \\
  f(v) & \mbox{otherwise}, 
 \end{array}\right. 
\end{eqnarray*}
 then $f_G^{(0)}$ clearly satisfies \eqref{eq:bipartite_feasibility}, and $f_G^{(2)}$ as well. 
 Thus, $G$ has a {\em connected} bipartite mod-4 $f_G^{(a)}$-factor $\Vec{y}^{(a)} \in \{0,1,2,3\}^{E(G)}$ for each $a \in \{0,2\}$ 
  since $G$ is bipartite all-round by the hypothesis. 
 If $(P,\Vec{x}^{(a)})$ consists of at most two connected component, 
  i.e., 
  $x^{(a)}_i = 0$ holds at most one $i \in \{1,2,\ldots,\ell\}$, 
  then $\Vec{x}^{(a)}$ and $\Vec{y}^{(a)}$ implies a connected mod-4 $f$-factor of $G'$.

 Then, we claim that 
  $(P,\Vec{x}^{(0)})$ or $(P,\Vec{x}^{(2)})$ consists of at most two connected component. 
 Remark that $\{ x_i^{(0)}, x_i^{(2)}\} = \{0,2\} \mbox{ or }\{1,3\}$ holds for each $i \in \{1,2,\ldots,\ell\}$. 
 Since $\ell \leq 3$, it is not difficult to observe the claim. 
\end{proof}

\subsubsection{Proof of Theorem~\ref{thm:p6_univ}}\label{sec:p6_univ_proof}
 Now, we show  Theorem~\ref{thm:p6_univ}. 
 In fact, 
   we prove the following lemma, 
   which with Lemma~\ref{crl:univ_all_edge} immediately implies Theorem~\ref{thm:p6_univ}. 
 \begin{lem} \label{prop:thm:p6_allround}
 Let $G=(V,E)$ be a two-edge connected $P_6$-free graph such that $G \neq C_5$. 
 For any edge $e \in E$, there is a subgraph $H$ of $G$
  such that $H$ contains $e$ and it is all-round or bipartite all-round. 
 \end{lem} 

\begin{proof}
 Suppose that an edge $e$ is contained in a cycle $C$ of length $\ell$. 
 If $\ell = 3,4,6$, 
  then $C$ itself is all-round or bipartite all-round, and we obtain the claim. 
 Thus, we will show the cases of $\ell=5$ and $\ell \geq 7$. 

 First, we show the case $\ell \geq 7$. 
 To begin with, we remark that $P_6$-free implies that 
  the cycle $C$ has at least two chords, say $f$ and $h$, which do not share their end vertices;
  otherwise, suppose they share a vertex $v$ then the induced subgraph removing $v$ contains $P_6$ (or longer). 
 For convenience, 
  we say a chord {\em separates $C$ into $(a, \ell - a)$} ($a \in \{ 2,3,\ldots,  \lfloor \ell/2 \rfloor \}$) 
  if the chord connects vertices in the distance $a$ along $C$. 
 The proof idea is an induction on $\ell$, 
  i.e., we show a shorter cycle containing~$e$. 
 Recall that 
  $C_3$ is all-round (Proposition~\ref{lem:k3_univ}), and 
  $C_4$ and $C_6$ are bipartite all-round (Proposition~\ref{lem:c4_c6_univ}), 
 while 
  $C_5$ is not all-round (Lemma~\ref{prop:C5_counterexample}).


 In the case of $\ell = 7$. 
 If a chord separates $C$ into $(2,5)$, 
  then $e$ is contained in $C_3$ or $C_6$, and we obtain the claim. 
 If both chords $f$ and $g$ separates $C$ into $(3,4)$, 
  then we claim that $H=C+f+g$ is all-round. 
 Note that $H$ is isomorphic to one of two graphs in Figure~\ref{fig:P6free-1}. 
 In the left graph, 
  we can find an ear decomposition consisting of 
  $C_4$ (1-5-6-7-1, in Figure~\ref{fig:P6free-1}) and ears of length 3 (1-2-3-7) and 2 (3-4-5). 
 Let $H'$ denote the graph consisting of $C_4$ (1-5-6-7-1) and an ear of length 3 (1-2-3-7), 
  then Lemma~\ref{lem:bip_less7_bip_all-round} implies that 
  $H'$ is bipartite all-round. 
 Since $H$ is given by $H'$ and an ear of length two, 
 Lemma~\ref{lem:bip_less3_all-round} implies that $H$ is all-round. 
 In the right graph, 
  we can find an ear decomposition consisting of 
  $C_4$ (1-5-6-7-1), an ear of length 2 (5,4,7), and an ear of length 3 (1,2,3,4). 
 Then, $H$ is also all-round by a similar argument. 

 In the case of $\ell=8$. 
 Then, a chord possibly separates $C$ into $(2,6)$, $(3,5)$ or $(4,4)$. 
 The cases of $(2,6)$ and $(3,5)$ are easy; 
 In the case of $(2,6)$, 
  $e$ is contained in $C_3$ or $C_7$, 
  where the latter case is reduced to the above case of $\ell=7$. 
 In the case of $(3,5)$, $e$ is contained in $C_4$ or $C_6$. 
 Finally, in the case that both $f$ and $g$ separates $C$ into $(4,4)$, 
  then $H=C+f+g$ is isomorphic one of two graphs in Figure~\ref{fig:P6free-2}. 
 One graph (upper one) consists of $C_4$ (1-5-4-8-1) and two ears of length 3 (1-2-3-4 and 5-6-7-8), and 
  it is  all-round by Lemma~\ref{lem:bip_less3_all-round}. 
 The other consists of $C_6$ (1-5-4-3-7-8-1) and ears of length 2 (1-2-3 and 5-6-7), and 
  it is bipartite all-round, 
  by Lemma~\ref{lem:bip_less7_bip_all-round}. 
 Thus we obtain the claim in the case.

 In the case of $\ell=9$.
 Then, a chord possibly separates $C$ into $(2,7)$, $(3,6)$ or $(4,5)$. 
 In the case of $(2,6)$ and $(3,5)$ are easy, 
   since $e$ is contained in $C_3$, $C_7$, $C_4$ or $C_6$. 
 If both $f$ and $g$ separates $C$ into $(4,5)$, 
  then $H=C+f+g$ is isomorphic one of three graphs in Figure~\ref{fig:P6free-3}. 
 We can observe that 
  each graph has an ear decomposition consisting of $C_6$ (1-2-3-4-5-6-1) with ears of length two and three. 
 Then, each graph is all-round 
  by Lemmas~\ref{lem:less7_univ}, \ref{lem:bip_less7_bip_all-round} and \ref{lem:bip_less3_all-round}.

 In the case of $\ell \geq 10$.
 Suppose that a chord separates $C$ into $(a,\ell - a)$ where $a \leq \ell-a$, then 
  $e$ is in $C_{a+1}$ or $C_{\ell-a+1}$. 
 Unless $e \in C_5$, i.e., $a=4$ and $e \in C_{a+1}$, the case is reduced to a shorter cycle. 
 Suppose $e \in C_{a+1}$ where $a=4$. 
 Then, $\ell-a+1 \geq 7$ implies that $C_{\ell-a+1}$ contains another chord. 
 Using the chord, we can reduce the case to a shorter cycle. 

 Next, we show the remaining case of $\ell=5$. 
 Suppose that $e$ is contained in a cycle $C$ of length five. 
 We consider two cases:
  $C$ is the unique cycle which contains $e$,
   or there is another cycle containing $e$.  
 
 First, we consider the former case.
 Since there is no cycle containing $e$ other than $C$, 
  $C$ has no chord. 
 Furthermore, since $G \neq C_5$, 
  there exists a vertex $v \in V(C)$ that is contained in another cycle,
  which implies there exists an edge $f \in \delta(v) \setminus E(C)$. 
 Since $G$ is a two-edge connected $P_6$-free graph,
  $f$ is contained in a $C_3$:
 Otherwise $G$ has a $P_6$ as an induced subgraph.
 Then, 
 Lemma~\ref{lem:less7_univ} implies that $C+C_3$ is all-round, and 
  we obtain the claim in the case. 
 
 Next, we consider the second case. 
 If there exists another cycle containing $e$ and its length is not five,
  the argument for $\ell \neq 5$ establishes Lemma~\ref{prop:thm:p6_allround} for $e$.
 Suppose every cycle containing $e$ has length five.
 Let $C''$ be one of the cycles containing $e$ other than $C$.
 Then $C$ and $C''$ has common edges.
 Let $k$ be the number of the common edges.
 If $k=2$ or $3$,
  Lemma~\ref{lem:bip_less3_all-round} implies that $C+C''$ is all-round. 
 If $k=1$,
  $C+C'' - e$ is a cycle of length eight.
 Then $C+C'' - e$ has a chord $f \neq e$
  because $G$ is $P_6$-free. 
 By the argument of the case $\ell=8$,
  $C+C''+f$ is all-round. 
 \end{proof}

 \begin{figure}[tbp]
	\begin{center}
		\includegraphics[width=130mm]{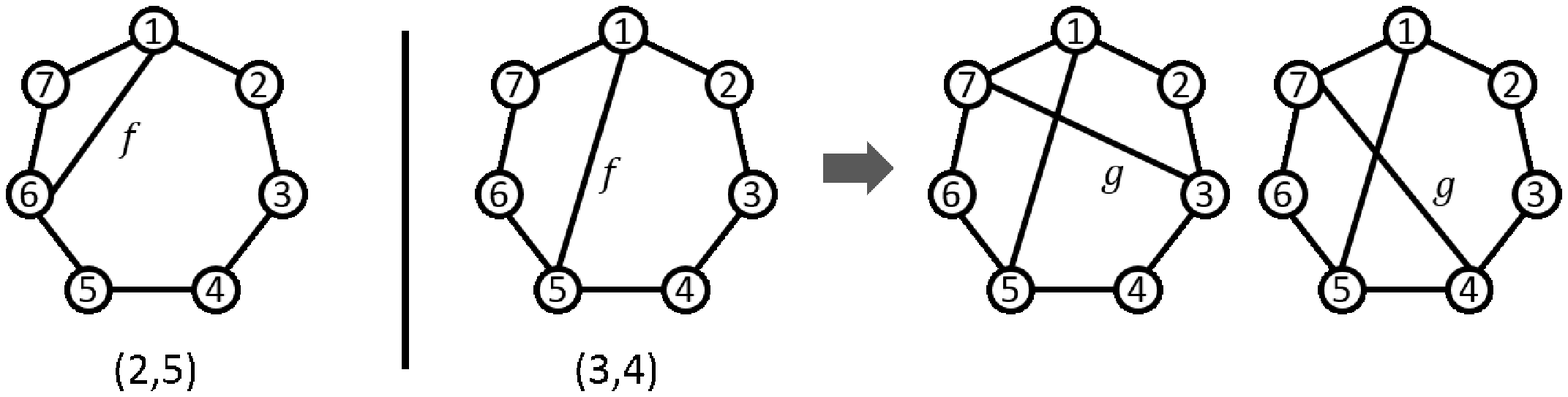}
	\end{center}
	\caption{$\ell = 7.$ \label{fig:P6free-1}}
\end{figure}
 \begin{figure}[tbp]
	\begin{center}
		\includegraphics[width=110mm]{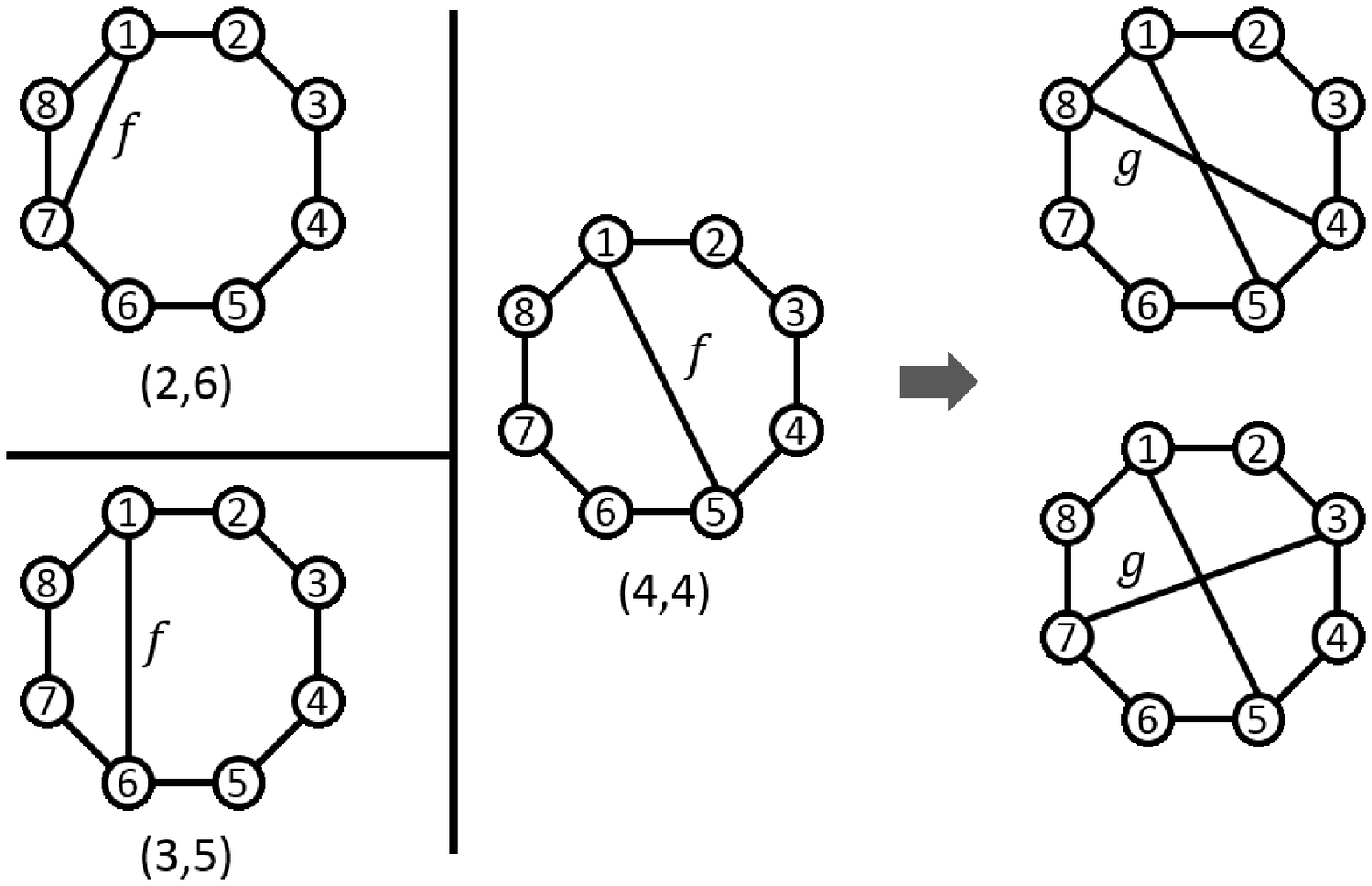}
	\end{center}
	\caption{$\ell = 8.$ \label{fig:P6free-2}}
\end{figure}
  \begin{figure}[tbp]
	\begin{center}
		\includegraphics[width=140mm]{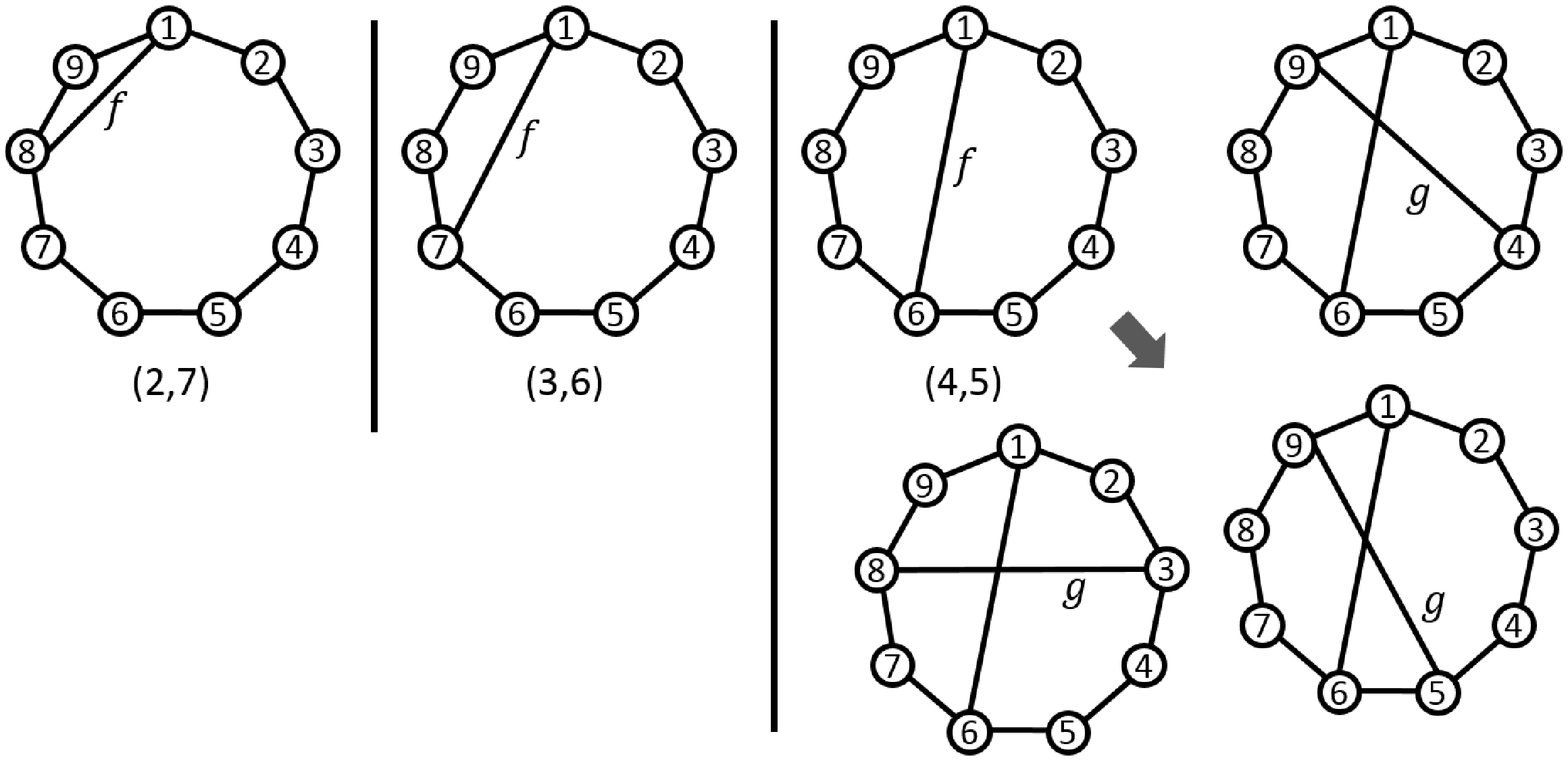}
	\end{center}
	\caption{$\ell = 9.$ \label{fig:P6free-3}}
\end{figure}

\section{Miscellaneous Discussions}\label{sec:discussion}
 This section remarks three related topics. 
 Section~\ref{sec:brdge} extends the arguments in Section~\ref{sec:PHC3} 
 to graphs with bridges, and 
 presents Theorem~\ref{thm:C5P6-connected}, which extends 
  Theorem~\ref{thm:PHC3-forC5P6free} 
  about two-edge connected graphs to any connected. 
 Section~\ref{sec:dense} remarks an all-roundness of dense graphs applying  Lemma~\ref{crl:univ_all_edge}. 
 Section~\ref{sec:connection} briefly explains a connection 
  between the PHC problem and other problems such as the HC problem. 

\subsection{All-roundness of graphs with bridges}\label{sec:brdge}
 This subsection is concerned with graphs with bridges. 
 \begin{thm}  \label{lem:suf_all-round_bridge}
 For a connected graph $G = (V,E)$, 
  let $B \subseteq E$ denote the set of bridges, and 
  let $S \subseteq V$ denote the set of isolated vertices in $G - B$.
 If $S = \emptyset$ and every two-edge connected component in $G - B$ is all-round 
  then $G$ is all-round. 
 \end{thm}
 \begin{proof}
 Let $C_1, \ldots, C_l$ be the two-edge connected components of $G - B$, and 
  let $H$ be a graph obtained by contracting every $C_i$ of $G$. 
 Notice that $E(H) = B$. 
 Let $u_1, \ldots, u_l$ be the vertices of $H$ 
  where each $u_i$ corresponds to $C_i$.
 For each $C_i$,
  let $T = \{u_i \mid \sum_{v \in C_i} f(v) \equiv 1 \pmod{2} \}$ and
  let $J$ be a $T$-join of $H$. 
 Then, let $\Vec{x} \in \{0,1,2,3\}^{E(H)}$ be defined by
 \begin{equation*}
 x_e = \left\{
 \begin{array}{ll}
	1 & \mbox{ if } e \in J, \\
	2 & \mbox{ if } e \notin J.
 \end{array}
 \right.
 \end{equation*}
 Let a map $f' \colon V \to \{0,1,2,3\}$ be defined by 
 \begin{equation*}
 f' (v) = \left( f(v) - \sum_{e \in \delta_H(v)} x_e \right ) \bmod 4
 \end{equation*}
  for each $v \in V(G)$. 
 Notice that $f' (v)=0$ for $v \in S$, 
  since $\Vec{x}$ is a mod-4 $f_H$-factor. 
 Since each $C_i$ is all-round, 
  $G$ has a mod-$4$ $f'$-factor $\Vec{y}$ which are connected in each two-edge connected component $C_i$. 
 Let $\Vec{x}' \in \{0,1,2,3\}^{E(G)}$ be defined by
 \begin{equation*}
 x'_e = \left\{
 \begin{array}{ll}
	x_e & \mbox{ if } e \in B, \\
	y_e & \mbox{ otherwise, }
 \end{array}
 \right.
 \end{equation*}
 then we 
  we obtain a connected mod-4 $f$-factor. 
 \end{proof}

 Now, we are concerned with the $\PHC_3$ problems again.  
 The following two propositions are easy observations from the fact that 
  any closed walk must pass each edge of $\delta(v)$ for any $v \in S$ an even number of times 
  since they are bridges. 
 \begin{prop}  \label{prop:nec_PHC_bridge}
 For a connected graph $G = (V,E)$, 
  let $B \subseteq E$ denote the set of bridges and 
  let $S \subseteq V$ denote the set of isolated vertices in $G - B$.
 If $G$ contains a ${\rm PHC}_3$, then $d_G(v)$ is odd for any $v \in S$.
\qed
 \end{prop}
 
 \begin{prop} \label{prop:suf_PHC_bridge}
 For a connected graph $G = (V,E)$, 
  let $B \subseteq E$ denote the set of bridges and 
  let $S \subseteq V$ denote the set of isolated vertices in $G - B$.
 Suppose that $d_G(v)$ is odd for any $v \in S$ and 
  that every connected component of $G - S$ is all-round, 
 then $G$ has a ${\rm PHC}_3$.
\qed
 \end{prop}
 
 By Theorem~\ref{lem:suf_all-round_bridge} and 
   Propositions~\ref{prop:nec_PHC_bridge} and \ref{prop:suf_PHC_bridge}, 
  we obtain the following theorem. 
\begin{thm}\label{thm:C5P6-connected} 
 Suppose that $G=(V,E)$ is a connected $P_6$-free, or $C_{\geq 5}$-free graph.
 Let $S \subseteq V$ denote the set of isolated vertex in $G - B$.
 Note that $S = \emptyset$ if $G$ is two-edge connected, while the reverse is not true. 
 Then, $G$ has a ${\rm PHC}_3$ if and only if 
  $d_G(v)$ is odd for every $v \in S$. 
\qed
\end{thm}

\subsection{All-roundness of dense graphs}\label{sec:dense}
 This section shows another application of Lemma~\ref{crl:univ_all_edge}. 
\begin{prop}  \label{thm:half_all-round}
 Let $G=(V,E)$ be a connected graph 
  where $|V| \geq 3$ and 
  the minimum degree of $G$ is at least $|V|/2$. 
 Then, $G$ is all-round, or bipartite all-round if $G$ is bipartite. 
\end{prop}
\begin{proof}
 We show that every edge $e = uv \in E$ is contained in a cycle of length three or four, 
  then Lemma~\ref{crl:univ_all_edge} implies that $G$ is all-round. 
 If $|V|$ is odd, 
  the degree of each $u$ and $v$ is strictly greater than $|V|/2$ by the hypothesis. 
 This implies that $u$ and $v$ has a common neighbor $w$ by the pigeon hole principle. 
 Thus any edge is contained in a cycle of length three. 
  
 Suppose $|V|$ is even. 
 If $u$ and $v$ has a common neighbor, then we obtain the claim. 
 If it is not the case, 
  we can observe that $|N(u) \setminus \{v\}|=|N(v) \setminus \{u\}|=|V|/2-1$ holds. 
 Let $w \in N(u) \setminus \{v\}$,
  then $d(w) \geq |V|/2$ implies that $w$ is connected to a vertex in $N(v) \setminus \{u\}$
  by the pigeon hole principle. 
 Thus we obtain a cycle of the length four in the case. 
\end{proof}

 In fact, we can show the following stronger theorem 
  with some complicated arguments (see Appendix~\ref{apdx:third_all-round} for the proof).  
\begin{thm}  \label{thm:third_all-round}
 Let $G=(V,E)$ be a connected graph 
  where $|V| \geq 4$ and 
  the minimum degree of $G$ is at least $|V|/3$. 
 Then, $G$ is all-round, or bipartite all-round if $G$ is bipartite. 
\end{thm}
 
\subsection{Connection of Parity Hamiltonian, Hamiltonian, Eulerian}\label{sec:connection}
We remark the connection between the PHC problem and the HC problem, or other related topics\footnote{
 The argument of this section might be appropriate appearing in the section of Concluding Remarks. 
 However, it is too long to put there, and we discuss just before the concluding remark. 
}. 
In fact, we are concerned with the following generalized version of Problem~\ref{prob:gen}. 
\begin{problem}[connected mod-$d$ $f$-factor with edge capacity constraints]\label{prob:gen2}
 Given a graph $G=(V,E)$, a map $z \colon E \to \mathbb{Z}_{\geq 0}$,  
 a positive integer $d$, and a map $f \colon V \to \mathbb{Z}_{\geq 0}$, 
  find $\Vec{x} \in \mathbb{Z}_{\geq 0}^E$ satisfying the conditions that 
\begin{align} 
 & \sum_{e \in \delta(v)} x_e \equiv f(v) \pmod{d} && \mbox{ for any $v \in V$, } \label{eq:modulo_constraint2} \\
 & \mbox{$(G, \Vec{x})$ is connected. } \label{eq:connectivity_constraint2} \\
 & x_e \leq z(e) && \mbox{ for any $e \in E$. }  \label{eq:capacity_constraint2} 
\end{align}
\end{problem}
 The PHC problem is given by setting $d=4$ and $f(v)=2$ for any $v \in V$ with an appropriate capacity constraint. 
 Problem~\ref{prob:gen} is given by setting $d=4$ and $z(e)=3$ for any $e \in E$. 
 The Hamiltonian cycle problem is represented by setting 
  $d=n$, $f(v)=2$ for any $v \in V$, and $z(e)=1$ for any $e \in E$. 
 We remark that 
  the Hamiltonian cycle problem is also given (in its original form) by setting 
  $d= \infty$, $f(v)=2$ for any $v \in V$, and 
  removing \eqref{eq:capacity_constraint2} (or setting $z(e)= \infty$ for any $e \in E$). 
 The Eulerian cycle problem is given by 
  setting $d=2$, $f(v)=0$ for any $v \in V$, and 
  replacing \eqref{eq:capacity_constraint2} with $x_e=1$ for any $e \in E$. 

 A two-factor plays a key role in the arguments of the HC problem in cubic graphs, 
  where the connectivity constraint is relaxed~\cite{Hartvigsen,HL11,BSvS,BIT13}. 
 Motivated by a {\em connected} ``factor,'' 
  this paper has investigated connected mod-$4$ factors. 
 A mod-$d$ factor for prime $d$ is an interesting future work.

\section{Concluding Remarks}
 In this paper, we have introduced the parity Hamiltonian cycle problem. 
 We have shown that the problem is in \PP\/ when $\z \geq 4$, 
  while {\NP}-complete when $\z \leq 3$. 
 Then, we are involved in the case $\z=3$, and 
   showed some graph classes for which the problem is in \PP. 
 It is open if the $\PHC_3$ problem is in \PP\ for three-edge connected graphs. 

 More sophisticated arguments on the connection between the PHC problem and related topics, 
  such as HC, $T$-join, even-factors, extended complexities, jump systems, etc., 
  are significant future works. 

\section*{Acknowledgements}
This work is partly supported by JSPS KAKENHI
Grant Number 15K15938, 25700002, 15H02666, and
Grant-in-Aid for Scientific Research on Innovative Areas
MEXT Japan Exploring the Limits of Computation (ELC).

\appendix 
\section{Finding A $T$-Join in Linear Time} \label{alg-Tjoin}
 In this section we describe an linear-time algorithm to find a $T$-join of a graph.
 The algorithm is described in Algorithm~\ref{alg:makeTjoin}.
 
\begin{algorithm}[tbhp]
 \caption{Linear Time Algorithm to Find a $T$-join}
 \label{alg:makeTjoin}
 \begin{algorithmic}
 \STATE MAKE-$T$-JOIN $(G,T)$
 \FORALL{$v \in V(G)$} \STATE $v.color \leftarrow {\rm WHITE}$ \ENDFOR
 \STATE $exchange\_flg \leftarrow {\rm FALSE}$
 \STATE $J \leftarrow \emptyset$
 \FORALL{$v \in V(G)$}
	\IF{$v.color = {\rm WHITE}$}
		\STATE MAKE-$T$-JOIN-REC $(G,T,v)$
		\IF{$exchange\_flg = {\rm TRUE}$} \STATE {\sf return} FALSE \ENDIF
	\ENDIF
 \ENDFOR
 \STATE {\sf return} $J$
 \end{algorithmic}
 \end{algorithm}
 
 \begin{algorithm}[thbp]
 \caption{Body of the Algorithm}
 \label{alg:makeTjoinRec}
 \begin{algorithmic}
 \STATE MAKE-$T$-JOIN-REC $(G,T,v)$
 \STATE $v.color \leftarrow {\rm BLACK}$
 \IF{$v \in T$} \STATE $exchange\_flg \leftarrow \overline{exchange\_flg}$ \ENDIF
 \FORALL{$u \in N(v)$}
	\IF{$u.color = {\rm WHITE}$}
		\IF{$exchange\_flg = {\rm TRUE}$} \STATE $J \leftarrow J \bigtriangleup \{u, v\}$ \ENDIF
		\STATE MAKE-$T$-JOIN-REC $(G,T,u)$
		\IF{$exchange\_flg = {\rm TRUE}$} \STATE $J \leftarrow J \bigtriangleup \{u, v\}$ \ENDIF
	\ENDIF
 \ENDFOR
 \end{algorithmic}
 \end{algorithm}
 
 The input of Algorithm~\ref{alg:makeTjoin} is
  a pair of an undirected graph $G$\footnote{ Here we does not assume that $G$ is connected.} 
  and a vertex set $T \subseteq V$.
 Algorithm~\ref{alg:makeTjoin} outputs a $T$-join $J$ if it exists,
  or {\rm FALSE} otherwise.

 Algorithm~\ref{alg:makeTjoin} essentially runs the depth first search,
  determining an edge should be picked for or removed from $J$
  while tracing edges.
 Algorithm~\ref{alg:makeTjoin} uses two variables, $color$ and $exchange\_flg$.
 $color$ is an attribute of each vertex $v$
  and it takes value of WHITE or BLACK.
 $v.color =$ WHITE indicates that $v$ is not visited by the algorithm,
 otherwise $v$ is already visited.
 $exchange\_flg$ is a boolean variable
  which determines an edge the algorithm is looking now should be picked for $J$ or not.
 If $exchange\_flg$ is TRUE the algorithm does the following;
  if $e \notin J$ then picks $e$ for $J$,
  otherwise removes $e$ from $J$.
 If $exchange\_flg$ is FALSE the algorithm does nothing for $e$.
 
 Since Algorithm~\ref{alg:makeTjoin} runs the depth first search with constant overheads,
  it clearly terminates in linear time.
 We show that the algorithm works correctly.
 
 \begin{prop} \label{prop:makeTjoin-correctness}
 Algorithm~\ref{alg:makeTjoin} returns a $T$-join of $G$ if it exists,
 or FALSE otherwise.
 \end{prop}
 \begin{proof}
 It suffices to show that
  Algorithm~\ref{alg:makeTjoin} works correctly on connected graphs;
  if $G$ has more than one connected components
  we can apply the correctness proof for each connected component.
 Suppose $G$ does not have a $T$-join, i.e., $|T|$ is odd.
 Since the number of inversion of $exchange\_flg$ is equal to $|T|$,
  the value of $exchange\_flg$ is inverted an odd number of times.
 Then Algorithm~\ref{alg:makeTjoinRec} ends the recursion with $exchange\_flg$ being TRUE,
  Algorithm~\ref{alg:makeTjoin} returns FALSE.
 
 Suppose $G$ has a $T$-join, i.e., $|T|$ is even.
 Let $H$ be a graph obtained by doubling every edge in $G$.
 We can regard the depth first search on $G$
  as a connected closed walk of $H$ which traces each edge exactly once.
 Let $e_1, e_2 \in E(H)$ be the copies of $e \in E(G)$,
  and let $J^\prime \subseteq E(H)$ be a set of edges
  which are passed when $exchange\_flg$ is TRUE.
 Then $e$ is a member of the output of Algorithm~\ref{alg:makeTjoin} $J$
  if and only if exactly one of $e_1$ and $e_2$ is a member of $J^\prime$, that is,
 \begin{equation} \label{eq:J_and_Jprime}
 J = \{ e \in E(G) \mid (e_1 \in J^\prime) \oplus (e_2 \in J^\prime) \}.
 \end{equation}
 The inversion of $exchange\_flg$ occurs if and only if the algorithm visits $v \in T$ first time,
  which causes $|\delta_H(v) \cap J^\prime|$ to be odd
  for each $v \in T$.
 On the other hand,
  it holds that $|\delta_H(v) \cap J^\prime|$ is even for every $v \notin T$,
  since $exchange\_flg$ is not inverted when the algorithm visits $v$.
 By \eqref{eq:J_and_Jprime}, 
  $|\delta_H(v) \cap J^\prime| \equiv |\delta_G(v) \cap J| \pmod{2}$ for each $v \in V(G)$;
  $|\delta_G(v) \cap J| \equiv 1 \pmod{2}$ for $v \in T$ and
  $|\delta_G(v) \cap J| \equiv 0 \pmod{2}$ for $v \notin T$.
 Thus $J$ is a $T$-join of $G$.
 \end{proof}

\section{Proof of Theorem~\ref{thm:third_all-round}} \label{apdx:third_all-round}
 To prove Theorem~\ref{thm:third_all-round},
  we need some more lemmas in the following. 
 
\begin{lem} \label{lem:univ_connect_graphs_ver1}
 Let $I$ be a graph given by connecting $G$ and $H$ by an edge,
  where $G$ and $H$ are all-round. 
 Then, $I$ is all-round.
\end{lem}
 \begin{proof}
 Let $e$ be the edges connecting $G$ and $H$.
 Notice that 
  $\sum_{v \in V(G)} f(v) \equiv \sum_{v \in V(H)} f(v) \pmod{2}$ since $f$ satisfies \eqref{eq:fiseven}. 
 Let $x_e = 1$ if $\sum_{v \in V(G)} f(v) \equiv 1 \pmod{2}$, 
  otherwise let $x_e = 2$. 
 Let $f^\prime$ be defined by \eqref{eq:f_prime}, 
  then $f'$ satisfies \eqref{eq:fiseven}. 
 Since $G$ and $H$ are all-round, $I$ has a connected mod-4 $f$-factor. 
\end{proof}

\begin{lem} \label{lem:univ_connect_graphs_ver2}
 Let $I$ be a graph given by connecting $G$ and $H$ by two edges,
  where $G$ is all-round and $H$ is bipartite all-round. 
 Then, $I$ is all-round.
\end{lem}
 \begin{proof}
 Let $e_1, e_2$ be the edges connecting $G$ and $H$.
 Set the value of $x_{e_1}$ and $x_{e_2}$ to satisfy the following two conditions:
 \begin{equation*}
 \sum_{v \in V(G)} f^\prime(v) \equiv 0 \pmod{2}
 \end{equation*}
 and
 \begin{equation*}
 \sum_{v \in U(H)} f^\prime(v) \equiv \sum_{v \in V(H)} f^\prime(v) \pmod{4}
 \end{equation*}
 where $f^\prime$ is defined by \eqref{eq:f_prime} and
  $U(H)$ and $V(H)$ is the color classes of $H$.
 Since $G$ is all-round and $H$ is bipartite all-round,
  $G \cup H$ has a connected mod-$4$ $f^\prime$-factor $x^*$.
 By a combination of $x^*$ and $x_{e_1}, x_{e_2}$, we obtain a connected mod-$4$ $f$-factor of $G$. 
\end{proof}

\begin{lem} \label{lem:univ_connect_graphs_ver3}
 Let $I$ be a graph given by connecting $G$ and $H$ by two edges,
  where $G$ and $H$ are bipartite all-round. 
 Then, $I$ is all-round or bipartite all-round.
\end{lem}
\begin{proof}
 Let $e_1, e_2$ be the edges connecting $G$ and $H$,
  and let $G=(A_1, B_1, E(G))$, $H=(A_2, B_2, E(H))$.
 There are three cases of which components $e_1$ and $e_2$ connect
 (see Figures~\ref{ver3_case1}, \ref{ver3_case2} and~\ref{ver3_case3}).

\medskip
\noindent{\bf Case 1.} 
If both $e_1$ and $e_2$ connect $A_1$ and $B_2$,
 set the value of $x_{e_1}, x_{e_2}$ to satisfy
\begin{equation}
\sum_{v \in A_1} f(v) - \sum_{v \in B_1} f(v) \equiv x_{e_1} + x_{e_2} \pmod{4}
\end{equation}
and at least one of $x_{e_1}$ or $x_{e_2}$ is not equal to zero.

\medskip
\noindent{\bf Case 2.} 
If $e_1$ connects $A_1$ and $B_2$,
 and $e_2$ connects $A_2$ and $B_1$,
 set the value of $x_{e_1}, x_{e_2}$ to satisfy
\begin{equation}
\sum_{v \in A_1} f(v) - \sum_{v \in B_1} f(v) \equiv x_{e_1} - x_{e_2} \pmod{4}
\end{equation}
and at least one of $x_{e_1}$ or $x_{e_2}$ is not equal to zero.

\medskip
\noindent{\bf Case 3.} 
If $e_1$ connects $A_1$ and $B_1$,
 and $e_2$ connects $A_1$ and $B_2$,
 set the value of $x_{e_1}, x_{e_2}$ to satisfy both
\begin{equation}
\sum_{v \in A_1} f(v) - \sum_{v \in B_1} f(v) \equiv x_{e_1} + x_{e_2} \pmod{4}
\end{equation}
and
\begin{equation}
\sum_{v \in A_2} f(v) - \sum_{v \in B_2} f(v) \equiv x_{e_1} - x_{e_2} \pmod{4},
\end{equation}
and at least one of $x_{e_1}$ or $x_{e_2}$ is not equal to zero.

Let $f^\prime$ be defined by \eqref{eq:f_prime},
 $G \cup H$ has a connected mod-$4$ $f^\prime$-factor $x^*$.
By combining $x^*$ and $x_{e_1}, x_{e_2}$, 
 we obtain a connected mod-$4$ $f$-factor of $G$.
\end{proof}

\begin{figure}[htbp]
	\begin{tabular}{cc}
		\begin{minipage}{0.3\textwidth}
			\begin{center}
				\includegraphics[width=40mm]{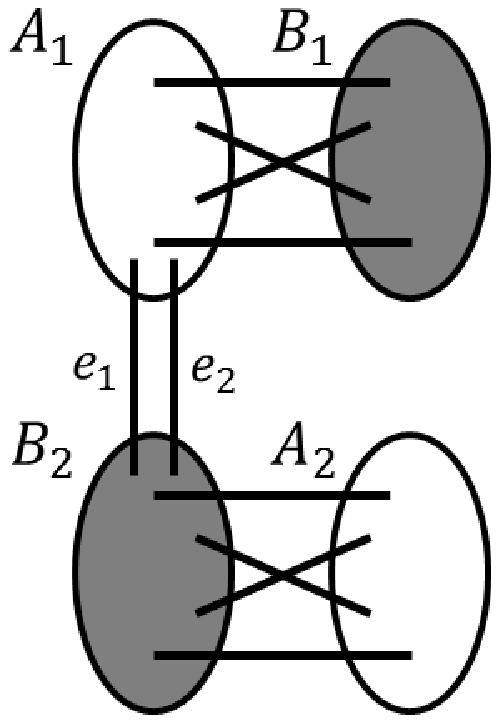}
				\caption{Case 1. \label{ver3_case1}}
			\end{center}
		\end{minipage}
		\begin{minipage}{0.3\textwidth}
			\begin{center}
				\includegraphics[width=40mm]{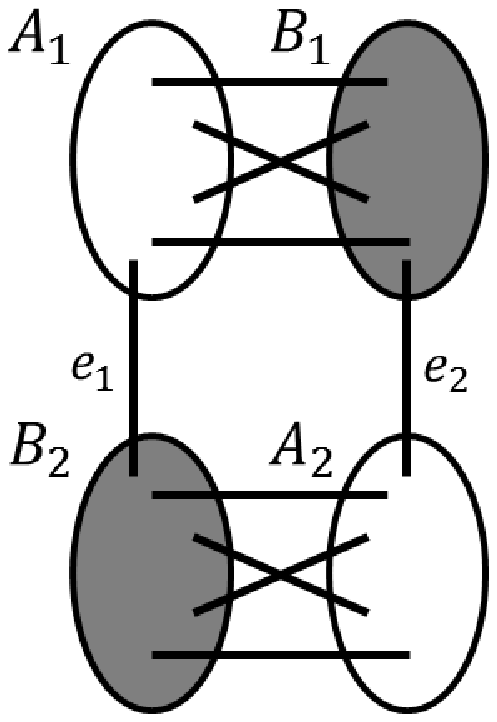}
				\caption{Case 2. \label{ver3_case2}}
			\end{center}
		\end{minipage}
		\begin{minipage}{0.3\textwidth}
			\begin{center}
				\includegraphics[width=40mm]{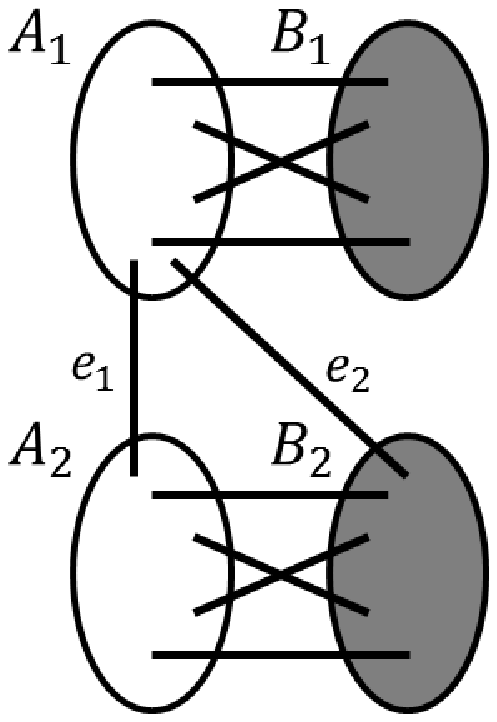}
				\caption{Case 3. \label{ver3_case3}}
			\end{center}
		\end{minipage}
	\end{tabular}
\end{figure}

 We also use the following two lemmas (we omit the proofs).
 
 \begin{lem} \label{lem:all_v_triangle_or_square}
 Let $G=(V,E)$ be a connected graph where $|V| \geq 8$ and 
  the minimum degree of $G$ is at least $|V|/3$.
 Then, every {\em vertex} in $G$ belongs to a triangle or a square.
 \end{lem}

 \begin{lem} \label{lem:exists_pair_with_2_bridges}
 Let $G=(V,E)$ be a connected graph where $|V| \geq 4$ and 
  the minimum degree of $G$ is at least $|V|/3$. 
 Let $C_1, \ldots, C_k$ ($k \geq 3$) be induced subgraphs of $G$
  such that $\bigcup_{i=1}^k V(C_i) = V$ and $V(C_i) \cap V(C_j) = \emptyset$ for any $i \neq j$. 
 Then there exists a pair $C_i, C_j (i \neq j)$ such that $|\delta(V(C_i)) \cap \delta(V(C_j))| \geq 2$. 
 \end{lem}
 
 Now we prove Theorem~\ref{thm:third_all-round}.
 
 
 \begin{proof} [Proof of Theorem~\ref{thm:third_all-round}]
 Let $n = |V|$.
 We can check 
  Theorem~\ref{thm:third_all-round} is true for any graph with $n \leq 7$ 
  by the brute force search. 
 Suppose $n \geq 8$.
 We consider two cases whether $G$ has a bridge or not\footnote{
  $G$ has at most one bridge since $\delta(G) \geq n/3$.}.
 
 \medskip
 \noindent{\bf Case 1.}
 $G$ has a bridge.

 Let $e$ be the bridge and
  $C_1, C_2$ be the two components of $G-e$
 (Note that neither $C_1$ nor $C_2$ is bipartite since $\delta(G) \geq n/3$).
 Then $|V(C_i)| \geq n/3 + 1$ for each $i = 1$ and $2$,
  therefore $n/3 + 1 \leq |V(C_i)| \leq (2/3)n - 1$ $(i=1,2)$.

 If $|V(C_i)| \leq (2/3)n - 2$ for both $i = 1,2$,
 \begin{equation*}
 \delta(C_i) \geq \frac{n}{3} - 1 = \frac{1}{2}\left( \frac{2}{3}n - 2\right) \geq \frac{1}{2}|V(C_i)|
 \end{equation*}
 holds for each $i = 1,2$,
  Theorem~\ref{thm:half_all-round} implies the all-roundness of $C_1$ and $C_2$.
 Thus, by Lemma~\ref{lem:univ_connect_graphs_ver1},
  $G$ is all-round.

 Suppose $|V(C_1)| = (2/3)n - 1$.
 Then $|V(C_2)| = (1/3)n+1 \leq (2/3)n - 2$, therefore $C_2$ is all-round as same as the previous case.
 Let $v \in V(C_1)$ be an end vertex of $e$.
 Then we have
 \begin{equation*}
 \delta(C_1 - v) \geq \frac{n}{3} - 1 = \frac{1}{2}\left( \frac{2}{3}n - 2\right) \geq \frac{1}{2}|V(C_i - v)|,
 \end{equation*}
  hence $C_1 - v$ is all-round.
 Since $C_1$ is obtained by adding ears of length at most two to $C_1 - v$,
  $C_1$ is all-round by Lemma~\ref{lem:less7_univ}.
 Thus $C_1$ and $C_2$ are both all-round,
  by Lemma~\ref{lem:univ_connect_graphs_ver1} $G$ is all-round.

 \medskip
 \noindent{\bf Case 2.}
 $G$ has no bridges.

 If $G$ has no bridges, $G$ is two-edge connected. 
 By Lemma~\ref{lem:all_v_triangle_or_square},
  every vertex of $G$ belongs to a triangle or a square.
 Let $C_1 ,\ldots, C_k$ be the triangles and squares,
  all-round (or bipartite-all-round) subgraphs of $G$.
 If $k = 2$,
  There are at least two edges connecting $C_1, C_2$ since $G$ is two-edge connected,
  thus $G$ is all-round by Lemmas~\ref{lem:univ_connect_graphs_ver1}, \ref{lem:univ_connect_graphs_ver2} and \ref{lem:univ_connect_graphs_ver3}.
 If $k \geq 3$,
  we can find a pair $C_i, C_j$ which has at least two connecting edges between them by Lemma~\ref{lem:exists_pair_with_2_bridges}.
 By combining $C_i$ and $C_j$ together with the connecting edges,
  we obtain a larger subgraph $C_i^\prime$ which is all-round or bipartite all-round,
  by Lemmas~\ref{lem:univ_connect_graphs_ver1}, \ref{lem:univ_connect_graphs_ver2} and \ref{lem:univ_connect_graphs_ver3}.
 This operation reduces the number of the subgraphs $k$ by one,
 by repeatedly taking this operation
 then we can reduce the case to $k = 2$.
 \end{proof}
\end{document}